\documentclass[10pt,conference]{IEEEtran}
\usepackage{graphicx}
\usepackage{amsthm}
\usepackage{epsfig}
\usepackage{latexsym}
\usepackage{amsfonts}
\usepackage{here}
\usepackage{rawfonts}
\usepackage[latin1]{inputenc}
\usepackage[T1]{fontenc}
\usepackage{calc}
\usepackage{capitalgreekitalic}
\usepackage{url}
\usepackage{enumerate}
\usepackage{color}
\usepackage[tbtags]{amsmath}
\usepackage{amssymb}
\usepackage{upref}
\usepackage{epic,eepic}
\usepackage{times}
\usepackage{dsfont}
\usepackage{comment}
\usepackage{cite}
\usepackage{mathrsfs}
\usepackage{verbatim}
\usepackage{float}
\usepackage{chngcntr}
\usepackage{bbm}
\usepackage{caption}
\usepackage{subcaption}
\allowdisplaybreaks
\usepackage{algorithm}
\usepackage{algpseudocode}
\DeclareMathOperator*{\argmax}{argmax}
\DeclareMathOperator*{\argmin}{argmin}

\newtheorem{theorem}{\bf Theorem}
\newtheorem{proposition}{\bf Proposition}
\let\oldproposition\proposition
\renewcommand{\proposition}{\oldproposition\normalfont}

\newtheorem{definition}{\bf Definition}
\let\olddefinition\definition
\renewcommand{\definition}{\olddefinition\normalfont}

\newlength{\totlinewidth}
\newcounter{substep}

  {\end{list}}

\newlength{\aligntop}
\newlength{\alignbot}
 \makeatletter
\renewenvironment{align}{%
  \vspace{\aligntop}
  \start@align\@ne\st@rredfalse\m@ne
}{%
  \math@cr \black@\totwidth@
  \egroup
  \ifingather@
    \restorealignstate@
    \egroup
    \nonumber
    \ifnum0=`{\fi\iffalse}\fi
  \else
    $$%
  \fi
  \ignorespacesafterend%
  \vspace{\alignbot}\par\noindent
} \makeatother

\IEEEoverridecommandlockouts

\begin{document}
\clearpage
\title{\huge Joint  Millimeter Wave and
	Microwave Resources Allocation in Cellular Networks with Dual-Mode Base Stations}
%
\author{{Omid Semiari$^{1}$}, Walid Saad$^{1}$, and Mehdi Bennis$^{2}$\\
\authorblockA{\small $^{1}$Wireless@VT, Bradley Department of Electrical and Computer Engineering, Virginia Tech, Blacksburg, VA, USA,\\
	 Email: \protect\url{osemiari@vt.edu}, \protect\url{walids@vt.edu} \\
$^{2}$Center for Wireless Communications, University of Oulu, Finland, Email: \protect\url{bennis@ee.oulu.fi}
    \thanks{This research was supported by the U.S. National Science Foundation under Grant CNS-1460316.}%
  }

%
%
%
%
\maketitle
\thispagestyle{empty}

\begin{abstract}
The use of dual-mode base stations that can jointly exploit millimeter wave (mmW) and microwave ($\mu$W) resources is a promising solution for overcoming the uncertainty of the mmW environment. In this paper, a novel dual-mode scheduling framework  is proposed that jointly performs user applications (UA) selection and scheduling over $\mu$W and mmW bands. The proposed scheduling framework  allows multiple UAs to run simultaneously on each user equipment (UE) and utilizes a set of \emph{context information}, including the channel state information (CSI) per UE, the delay tolerance and required load per UA, and the uncertainty of mmW channels, to maximize the quality-of-service (QoS) per UA. The dual-mode scheduling problem is then formulated as an optimization with minimum unsatisfied relations (min-UR) problem which is shown to be challenging to solve. Consequently, a long-term scheduling framework, consisting of two stages, is proposed. Within this framework, first, the joint UA selection and scheduling over $\mu$W band is formulated as a one-to-many matching game between the $\mu$W resources and UAs. To solve this problem, a novel scheduling algorithm is proposed and shown to yield a two-sided stable resource allocation. Second, over the mmW band, the joint context-aware UA selection and scheduling problem is formulated as a 0-1 Knapsack problem and a novel algorithm that builds on the Q-learning algorithm is proposed to find a suitable mmW scheduling policy while adaptively learning the UEs' line-of-sight probabilities. Furthermore, it is shown that the proposed scheduling framework can find an effective scheduling solution, over both $\mu$W and mmW, in polynomial time. Simulation results show that, compared with conventional scheduling schemes, the proposed approach significantly increases the number of satisfied UAs while improving the statistics of QoS violations and enhancing the overall users' quality-of-experience.
\end{abstract}
\vspace{-0cm}
\section{Introduction} \vspace{-0cm}
Communication at high frequency, millimeter wave (mmW) bands is seen as promising approach to overcome the scarcity of the radio spectrum while providing significant capacity gains for tomorrow's wireless cellular networks \cite{6736746,Rangan14,Ghosh14}. 
However, field measurements \cite{Rangan14} have shown that the availability of mmW links can be highly intermittent, due to various factors such as blockage by different obstacles. Therefore, meeting the stringent quality-of-service (QoS) constraints of delay-sensitive applications, such as HDTV and video conferencing, becomes more challenging at mmW frequencies compared to sub-6 GHz frequencies~\cite{shokri,7381698,7010536,5783993, 4689210,7010539,Rangan14}.

Such strict requirements can be achieved by deploying dual-mode small base stations (SBSs) that can support high data rates and QoS by leveraging the available bandwidth at both mmW and microwave ($\mu$W) frequency bands \cite{5783993}. Indeed, in order to provide robust and reliable communications, mmW networks must coexist with small cell LTE networks that operate at the conventional $\mu$W band \cite{7010536,5783993, 4689210,7010539,7430349}. However, differences in signal propagation characteristics  and in the available bandwidth lead to a significant difference in the achievable rate and the QoS over mmW and $\mu$W frequency bands, thus, yielding new challenges for joint mmW-$\mu$W user scheduling \cite{Wells,7430349}. In addition, QoS provisioning in dual-mode mmW-$\mu$W networks requires overcoming two key challenges:
1) a joint scheduling over both frequency bands is required, since resource allocation over one band will affect the allocation of the resources over the other frequency band and 2) the QoS constraints per user application (UA) will naturally dictate whether the traffic should be served via mmW resources, $\mu$W resources, or both. Therefore, robust and efficient scheduling algorithms for dual-mode SBSs are required that exploit \emph{context information per UA}, including the channel state information (CSI), maximum tolerable delay, and the required load to maximize users' quality-of-experience (QoE).

To address these scheduling challenges, a number of recent works have emerged in \cite{5783993,4689210,7010536,7381698,7010539,7430349} and \cite{Mezzavilla16,Singh16,Wang16,omid16,7430349,7572076}. 
The work in \cite{5783993} provides an overview on possible mmW-$\mu$W dual-mode architectures that can be used to transmit control and data signals, respectively, at $\mu$W and mmW frequency bands. To cope with the intermittent mmW link quality, the authors in \cite{Mezzavilla16} formulate the handover decision problem as a Markov decision process (MDP) in mmW networks. In addition, the work in \cite{Singh16} studies the problem of radio access technology (RAT) selection and traffic aggregation where each user can simultaneously be connected to multiple BSs. In \cite{Wang16}, the authors develop an RAT selection scheme for mmW-$\mu$W networks via a multi-armed bandit problem that aims to minimize the cost of handoffs for the UEs. \textcolor{black}{Furthermore, the authors in \cite{7572076} propose a cross-layer resource allocation scheme for full-duplex communications at the 60 GHz mmW frequency band.}

Although interesting, the body of work in \cite{5783993}, \cite{Mezzavilla16,Singh16,Wang16}, and \cite{7572076} does not address the scheduling problem in mmW-$\mu$W networks. In fact, \cite{Mezzavilla16,Singh16,Wang16} focus only on the cell association problem without taking into account, explicitly, the joint allocation of mmW and $\mu$W resources. Moreover, existing works such as in \cite{Mezzavilla16} and \cite{7572076} are solely focused on the mmW network, while completely neglecting the impact of the communications over the $\mu$W frequencies.    

In \cite{7381698}, the authors propose an energy-efficient resource allocation scheme for cellular networks, leveraging both $\mu$W and unlicensed  $60$ GHz mmW bands. \textcolor{black}{In \cite{7430349}, the resource allocation problem for ultra-dense mmW-$\mu$W cellular networks is studied under a model in which the cell association is decoupled in the uplink for mmW users. However, this work does not consider any QoS constraint in mmW-$\mu$W networks.} The problem of QoS provisioning for mmW networks is studied in \cite{7010536,4689210,7010539}, and \cite{omid16}. In \cite{7010536}, the authors propose a scheduling scheme that integrates device-to-device mmW links with 4G system to bypass the blocked mmW links. The work in \cite{4689210} presents a mmW system at $60$ GHz for supporting uncompressed high-definition (HD) videos for WLANs. In \cite{7010539}, the authors evaluated key metrics to characterize multimedia QoS, and designed a QoS-aware multimedia scheduling scheme to achieve the trade-off between performance and complexity.

Nonetheless, \cite{4689210} and \cite{7010539} do not consider multi-user scheduling and multiple access in dual-mode networks. In addition, conventional scheduling mechanisms, such as in \cite{7381698,7010536,4689210}, and \cite{7010539}, identify each UE by a single traffic stream with a certain QoS requirement. In practice, however, recent trends have shown that users run multiple applications simultaneously, each with a different QoS requirement. \textcolor{black}{Although the applications at a single device experience the same wireless channel, they may have different QoS requirements. For example, the QoS requirements for an interactive video call are more stringent than updating a background application or downloading a file. With this in mind, each user's QoE must be defined as a function of the number of QoS-satisfied UAs}. Accounting for precise, application-specific QoS metrics is particularly important for scheduling mmW resources whose channel is highly variable, due to large Doppler spreads and short channel coherence time. In fact, conventional scheduling approaches fail to guarantee the QoS for multiple applications at a single UE.

In \cite{omid16}, we studied the problem of the resource management with QoS constraints for dual-mode mmW-$\mu$W SBSs. However, our previous work assumes that mmW resources are allocated to the users opportunistically and the SBS does not dynamically determine the probability of being at a line-of-sight (LoS) link with the UEs. Such a resource allocation approach may not be efficient, since it may schedule users with low probability of LoS over the mmW band which effectively reduces the available resources for the users with high LoS probability. 
In addition, we will extend the results provided in \cite{omid16} by proposing a Q-learning model to dynamically learn the LoS probability per user and exploit this information to enhance the scheduling over mmW band.

\textcolor{black}{The main contribution of this paper is to propose a novel, context-aware scheduling framework for enabling a dual-band base station to jointly and efficiently allocate both mmW and $\mu$W resources to user applications. This proposed context-aware
scheduler allows each user to seamlessly run multiple applications simultaneously, each having its own distinct QoS constraint. To this end, the proposed scheduling problem is formulated as an optimization problem with minimum unsatisfied relations (min-UR) and the goal is to maximize the number of satisfied user applications (UAs). To solve this NP-hard problem, a novel scheduling framework is proposed that considers a set of context information composed of the UAs' tolerable delay, required load, and the line-of-sight (LoS) probability, to jointly select and schedule UAs over the $\mu$W and the mmW frequency bands.
 The resource allocation problem at $\mu$W band is modeled as a matching game that aims to assign resource blocks (RBs) to the candidate UAs. To solve this game, a novel algorithm is proposed that iteratively solves the UA selection and the resource allocation problems. We show that the proposed algorithm is guaranteed to yield a two-sided stable matching between UAs and the $\mu$W RBs. Over the mmW band, the scheduler assigns priority, based on the context information, to the remaining UAs that were not scheduled over the $\mu$W band. 
 Consequently, over the mmW band, we show that the scheduling problem can be cast as a 0-1 Knapsack problem. To solve this problem, we then propose a novel algorithm that allocates the mmW resources to the selected UAs. Moreover, we show that the proposed, two-stage scheduling framework can solve the context-aware dual-band scheduling problem in polynomial time with respect to the number of UAs. Simulation results show that the proposed approach significantly improves the QoS per application, compared to the proportional fair and round robin schedulers. }

The rest of this paper is organized as follows. Section \ref{sec:model} presents the problem formulation. Section \ref{Sec:III} presents the proposed context-aware scheduling solution over the $\mu$W band. The proposed context-aware scheduling solution over the mmW band is proposed in Section \ref{Sec:IV}. Simulation results are analyzed in Section \ref{Sec:V}. Section \ref{Sec:VI} concludes the paper.   
\begin{table*}[!t]
\footnotesize
\centering
\caption{Variables and notations}\vspace*{-0.2cm}
\begin{tabular}{|c|c||c|c|}
\hline
\bf{Notation} & \bf{Description} & \bf{Notation} & \bf{Description} \\
\hline
$M$& Number of UEs & $\mathcal{M}$ & Set of UEs\\
\hline
$\kappa_m$& Number of UAs per UE $m$& A& Total number of UAs\\
\hline
$\tau$& Time slot duration& $\tau'$& Beam training overhead\\
\hline
$L$& Path loss& $\pi_t$& Scheduling decision at time slot $t$\\
\hline
$K_1$& Number of $\mu$W RBs & $\mathcal{K}_1$& Set of $\mu$W RBs\\
\hline
$K_2$& Number of mmW RBs & $\mathcal{K}_2$& Set of mmW RBs\\
\hline
$\mathcal{G}_{t,1}$& Set of UAs to be scheduled over $\mu$W band& $\mathcal{G}_{t,2}$& Set of UAs to be scheduled over mmW band\\
\hline
$w_1$& Bandwidth of $\mu$W RBs& $w_2$& Bandwidth of mmW RBs\\
\hline
$g_{kt}$& $\mu$W channel over RB $k$ at time slot $t$& $\zeta \in \{0,1\}$& $\zeta=1$ if link is LoS, otherwise, $\zeta=0$. \\
\hline
$h_{kt}$& mmW channel over RB $k$ at time slot $t$& $p_{k,1}$& Transmit power over $\mu$W RB $k$\\
\hline
$\rho_a$&  LoS probability of the link for UA $a$& $p_{k,2}$& Transmit power over mmW RB $k$\\
\hline
$P_1$& Total transmit power over $\mu$W band& $P_2$& Total transmit power over mmW band \\
\hline
$\boldsymbol{r}$& Q-learning reward vector& $J$& Number of QoS classes \\
\hline
$\mathcal{A}$& Set of all UAs across all UEs& $\mathcal{A}_j$& Set of UAs with $j$-th QoS class\\
\hline
$b_a$& Total required bits for UA $a$& $b_a^{\textrm{rec}}(t)$& Received bits by UA $a$ during time slot $t$ \\
\hline
$\lambda_{t,1}$& Number of satisfied UAs over $\mu$W band& $\lambda_{t,2}$& Number of satisfied UAs over mmW band\\
\hline
\end{tabular}\label{tab1} \vspace{-2em}
\end{table*}

\section{System Model}\label{sec:model}
Consider the downlink of a dual-mode SBS that operates over both $\mu$W and mmW frequency bands. The coverage area of the SBS is a planar area with radius $d$ centered at $(0,0) \in \mathbb{R}^2$. Moreover, a set $\mathcal{M}$ of $M$ UEs is deployed randomly and uniformly within the SBS coverage. UEs are equipped with both mmW and $\mu$W RF interfaces which allow them to manage their traffic at both frequency bands \textcolor{black}{\cite{pi2011techniques}}. 
The antenna arrays of mmW transceivers can achieve an overall beamforming gain of $\psi(y_1,y_2)$ for a LoS UE located at $(y_1,y_2) \in \mathbb{R}^2$ \cite{Ghosh14}. \textcolor{black}{Meanwhile, the $\mu$W transceivers have conventional single element, omni-directional antennas to maintain low overhead and complexity at the $\mu$W frequency band \cite{5378666}}. In our model, each UE $m \in \mathcal{M}$ runs $\kappa_m$ UAs. We let $\mathcal{A}$ be the set of all UAs with $A=\sum_{m \in \mathcal{M}}\kappa_m$ as the total number of UAs across all UEs. \vspace{-1em}
\subsection{Channel Model and Multiple Access at mmW and $\mu$W Frequency Bands}
The downlink transmission time is divided into time slots of duration $\tau$. 
For the $\mu$W band, we consider an orthogonal frequency division multiple access (OFDMA) scheme in which multiple UAs can be scheduled over $K_1$ resource blocks (RBs) in the set $\mathcal{K}_1$ at each time slot with duration $\tau$. Therefore, the achievable \emph{$\mu$W rate} for an arbitrary UA $a$ at RB $k$ and time slot $t$ is:\vspace{.2cm}
\begin{align}\label{rate2}
R_a(k,t)=w_1\log_2\left(1+\frac{p_{k,1}|g_{kt}|^2 10^{-\frac{L_1(y_1,y_2)}{10}}}{w_1N_0}\right).
\end{align}
Here, $w_1$ is the bandwidth of each RB at $\mu$W band, and $g_{kt}$ is the Rayleigh fading channel coefficient over RB $k$ at time slot $t$. The total transmit power at $\mu$W band, $P_1$, is assumed to be distributed uniformly among all RBs such that $p_{k,1}=P_1/K_1$. \textcolor{black}{This uniform power allocation assumption is
	due to the fact that at a high SNR regime, as is expected in small cells with relatively
	short-range links, optimal power allocation
	policies such as the popular water-filling algorithm will ultimately converge to a uniform power
	allocation \cite{5733445}}. The path loss $L_1(y_1,y_2)$ follows the log-distance model with parameters $\alpha_1$, $\beta_1$, and $\xi_1^2$ that represent, respectively, the path loss exponent, the path loss at $1$ meter distance, and the variance of the shadowing for the $\mu$W band.  


\textcolor{black}{Over the mmW band, directional transmissions are inevitable to overcome the significantly high path loss at the mmW frequencies. Therefore, the multiple access scheme at the mmW band should support directional transmissions, while maintaining low complex designs for transceivers. Thus, the SBS uses a time division multiple access (TDMA) scheme to schedule UAs \cite{Rangan14}, which is in line with the existing standards such as WirelessHD and IEEE 802.15.3c \cite{5936164}}. We let $\mathcal{G}_{t,2}$ be the set of UAs that must be scheduled over the mmW band at slot $t$. During each time slot, for UAs that are assigned to the mmW band, the SBS transmits to UA $a \in \mathcal{G}_{t,2}$ an OFDM symbol of duration $\tau_{a,t}$ composed of $K_2$ resource blocks (RBs). In practice, the mmW transceivers must align their beams during a \emph{beam training} phase, in order to achieve the maximum beamforming gain \cite{7218630}. This training phase will introduce a non-negligible overhead on the TDMA system, which can become particularly significant as the number of mmW users increases. Hence, a beam training overhead time $\tau'<\tau$ is considered per transmission to a UA over the mmW band. In practice, duration of $\tau'$ can reach up to $1.54$ milliseconds, depending on the beam resolution \cite{7218630}. Therefore, the effective time for data transmission to UAs in $\mathcal{G}_{t,2}$ will be $\sum_{a \in \mathcal{G}_{t,2}}\tau_{a,t}=\tau-|\mathcal{G}_{t,2}|\tau'$, where $|\mathcal{G}_{t,2}|$ denotes the cardinality of the set $\mathcal{G}_{t,2}$.

The large-scale channel effects over the mmW links follow the popular model of \cite{Ghosh14}:
\begin{align}\label{mwpathloss}
L_2(y_1,y_2)=\beta_2+\alpha_2 10\log_{10}(\sqrt{y_1^2+y_2^2})+\chi,
\end{align}
where $L_2(y_1,y_2)$ is the path loss at mmW frequencies for all UAs associated with a UE located at $(y_1,y_2) \in \mathbb{R}^2$. In fact, (\ref{mwpathloss}) is known to be the best linear fit to the propagation measurement in mmW frequency band\cite{Ghosh14}, where $\alpha_2$ is the slope of the fit and $\beta_2$, the intercept parameter, is the pathloss (dB) for $1$ meter of distance. In addition, $\chi$ models the deviation in fitting (dB) which is a Gaussian random variable with zero mean and variance $\xi_2^2$. Overall, the total achievable \emph{mmW rate} for UA $a$ at time slot $t$ is given by  
	\begin{align}\label{rate1}
	&R_a(t)=\notag\\
	&\begin{cases}
	\sum_{k=1}^{K_2}\!w_2\log_2\!\left(\!\!1\!+\!\frac{p_{k,2}\psi(y_1,y_2)|h_{kt}|^2 10^{-\frac{L_2(y_1,y_2)}{10}}}{w_2 N_0}\right), & \zeta_{at}=1,\\
	0, &\zeta_{at}=0,
	\end{cases}
	\end{align}
where $\zeta_{at}=1$ indicates that a LoS link is feasible for UA $a$, otherwise, $\zeta_{at}=0$ and the link is blocked by an obstacle. In fact, $\zeta_{at}$ is a Bernoulli random variable with probability of success $\rho_{a}$, and is identical for all UAs that are run by the same UE. Moreover, $w_2$ is the bandwidth of each RB, \textcolor{black}{$h_{kt}$ is the Rician fading channel coefficient at RB $k$ of slot $t$ \cite{7481410}}, and $N_0$ is the noise power spectral density. Furthermore, $p_{k,2}$ denotes the SBS transmit power at RB $k$ of mmW frequency band. The total transmit power at mmW band, $P_2$, is assumed to be distributed uniformly among all RBs, such that $p_{k,2}=P_2/K_2$.
\begin{figure}
\centering
\centerline{\includegraphics[width=\columnwidth]{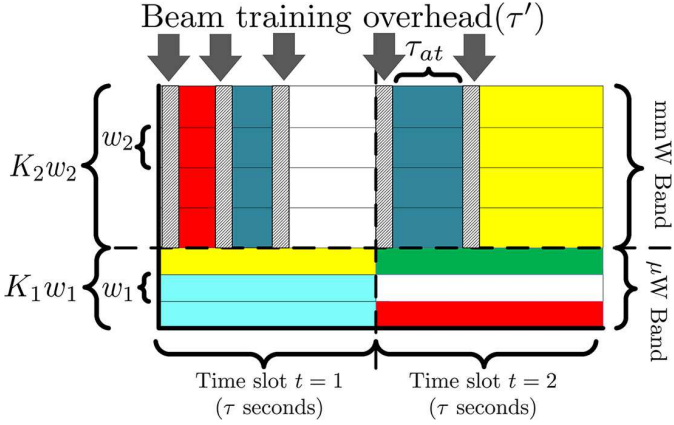}}\vspace{-0cm}
\caption{\small Example of resource allocation of the dual-band configuration. Colors correspond to different UAs that may run at different UEs.}\vspace{-0em}
\label{model}
\end{figure}


Let $\mathcal{G}_{t,1}$ be the set of UAs that must be scheduled over the $\mu$W band at time slot $t$. During each time slot, a UA can be scheduled only at one frequency band, i.e., $\mathcal{G}_{t,1} \cap \mathcal{G}_{t,2}= \emptyset$. 

The proposed dual-band multiple access scheme is shown in Fig. \ref{model}, where each color identifies a single, distinct UA. \vspace{-1em}  
\subsection{Traffic Model with QoS Constraints}
We assume a non-full buffer traffic model in which an arbitrary UA $a$ has a total of $b_a$ bits of data to receive. In addition, each UA has an application-specific tolerable delay which specifies its QoS class, as formally defined next.
\begin{definition}
The \textit{QoS class}, $\mathcal{A}_j$, is defined as the set of all UAs stemming from all UEs that can tolerate a maximum packet transmission delay of $j$ time slots.
\end{definition} 

Each UA in our system can belong to one out of a total of $J$ QoS classes, $\mathcal{A}_j, j=1,\cdots,J$ with $\bigcup_{j=1}^J \mathcal{A}_j=\mathcal{A}$, and $\mathcal{A}_j \cap \mathcal{A}_{j'}=\emptyset, j \neq j'$. Due to system resource constraints, not all UAs can be served instantaneously and, thus, a transmission delay will be experienced by the UAs.  In essence, to transmit a data stream of size $b_a$ bits to UA $a \in \mathcal{A}_j$, an average data rate of $b_a/j\tau_2$ during $j$ consecutive time slots is needed, otherwise, the UA experiences an outage due to the excessive delay. 

The scheduling decision $\pi_t$ at a given slot $t$ is a function that outputs two vectors $\boldsymbol{x}_t$ and $\boldsymbol{\tau}_t$ that determine, respectively, the resource allocation over $\mu$W and mmW bands. In fact, $\boldsymbol{x}_t$ includes the variables $x_{akt}\in \{0,1\}$ with $a \in \mathcal{A}, k \in \mathcal{K}_1$ where $x_{akt}=1$ indicates that $\mu$W RB $k$ is allocated to UA $a$ at slot $t$, otherwise, $x_{akt}=0$. In addition, each element $\tau_{at} \in [0,\tau], a \in \mathcal{A}$, of $\boldsymbol{\tau}_t$ determines the allocated time to UA $a$ over mmW band. The required bits for UA $a$ at slot $t$, $b_{a}^{\textrm{req}}(t)$, depend on the number of bits received during previous slots, $\sum_{t'=0}^{t-1}b_{a}^{\textrm{rec}}(t')$. In other words, $b_{a}^{\textrm{req}}(t)=b_a-\sum_{t'=0}^{t-1}b_{a}^{\textrm{rec}}(t')$, with $b_{a}^{\textrm{req}}(1)=b_a$ and $b_a^{\textrm{rec}}(0)=0$.  For a given policy $\pi_{t}$, the required load at time slot $t+1$, $b_{a}^{\textrm{req}}(t+1)$, can be written recursively as
\begin{align}\label{eqRbar}
b_{a}^{\textrm{req}}(t+1)&=b_{a}^{\textrm{req}}(t)-b_a^{\textrm{rec}}(t),\\\notag
&=b_{a}^{\textrm{req}}(t)-\left[\tau\sum\limits_{k=1}^{K_1} R_a(k,t)x_{akt}+ R_a(t)\tau_{at}\zeta_{at}\right].
\end{align}
From \eqref{eqRbar}, we observe that policy $\pi_t$ depends on the scheduling decisions during previous time slots $\{\pi_1, \pi_2, ..., \pi_{t-1}\}$. Thus, we define $\pi=\{\pi_1, \pi_2, ..., \pi_t, ...,\pi_J \} \in \boldsymbol{\Pi}$ as a long-term \textit{scheduling policy}, where $\boldsymbol{\Pi}$ is the set of all possible scheduling policies.

Next, we use (\ref{eqRbar}) to formally define the \emph{QoS criterion} for any UA $a \in \mathcal{A}_j$ as 
\begin{align}\label{lambda2}
\mathds{1}(a\in \mathcal{A}_j;\pi) = \begin{cases}
1             & \text{if}\,\, \sum_{t'=1}^{j}b_a^{\textrm{rec}}(t')\geq b_a,\\
0             & \text{otherwise},
\end{cases}
\end{align}
where $\mathds{1}(a\in \mathcal{A}_j;\pi)=1$ indicates that under policy $\pi$, enough resources are allocated to UA $a \in \mathcal{A}_j$ to receive $b_a$ bits within $j$ slots, while $\mathds{1}(a\in \mathcal{A}_j;\pi)=0$ indicates that UA $a$ is going to experience an outage. We define the outage set $\mathcal{O}^{\pi}=\{a| \mathds{1}(a\in \mathcal{A}_j;\pi)=0, j=1,\cdots,J\}$ as the set of UAs in outage.

Prior to formulating the problem, we must note the following inherent characteristics  of dual-mode scheduling: 1) If mmW link with a high LoS probability is not feasible for a UE, scheduling over the mmW band can cause outage to the associated UAs, specifically for delay-intolerant UAs, 2) larger range of supported rates is available for UAs compared to the conventional single-band systems. Hence, for some UAs, the required rate exceeds the achievable rate at $\mu$W band. Therefore, effective dual-mode scheduling should not only rely solely on CSI, but it must also leverage UA-specific metrics, herein referred to as \textit{context information} as formally defined next. 
\begin{definition}
At any slot $t$, the tuple $\mathcal{C}=(\mathcal{A}_{j\geq t}, \boldsymbol{b}^{\textrm{req}}(t), \boldsymbol{\rho})$ defined as \textit{context information}, is composed of the delay constraints of UAs, $\mathcal{A}_{j\geq t}=\bigcup_{j=t}^J \mathcal{A}_{j}$, the required load per UA, $\boldsymbol{b}^{\textrm{req}}(t)=\{b_{a}^{\textrm{req}}(t)| a \in \mathcal{A}_{j\geq t}\}$, and the LoS probability of each UA, $\boldsymbol{\rho}=\{\rho_{a}| a \in \mathcal{A}_{j\geq t}\}$.
\end{definition}
Note that exploiting the context information at any time slot $t$ properly links the scheduling policy $\pi_t$ to the history, since from \eqref{lambda2}, $\mathcal{C}$ at slot $t$ depends on $\pi_t', t'=1,\cdots,t-1$.   
\subsection{Problem Formulation}
Our goal is to find a scheduling policy $\pi^* \in \boldsymbol{\Pi}$ that satisfies \eqref{lambda2} for as many UAs as possible over $J$ time slots. 
The general long-term scheduling problem for slots $t=1,\cdots,J$ can be solved separately  at each slot $t$ to find $\pi_t^*(\mathcal{C},\text{CSI})=(\boldsymbol{x}^*_t, \boldsymbol{\tau}^*_t)$, while the time-dependency of scheduling decisions is captured by exploiting the context information. Therefore, the scheduling problem at an arbitrary slot $t$ can be formulated as follows:\vspace{-.5em}
\begin{subequations}
\begin{IEEEeqnarray}{l}
 \argmax_{\boldsymbol{x}_t,\boldsymbol{\tau}_t} \lambda_{t,1}+\mathbb{E}\left[\lambda_{t,2}\right],\label{opt1:a}\\
\text{s.t.}\,\,\,\,\,
\tau\sum\limits_{k=1}^{K_1} R_a(k,t)x_{akt} \geq b_{a}^{\textrm{req}}(t),  \forall a \in \mathcal{A}_t \cap \mathcal{G}_{t,1},\label{opt1:b}\\
R_a(t)\tau_{at}\zeta_{at} \geq b_{a}^{\textrm{req}}(t-1),  \forall a \in \mathcal{A}_t \cap \mathcal{G}_{t,2},\label{opt1:c}\\
 \boldsymbol{x}_t\in \mathcal{X}=\\\notag
 \left\{x_{akt} \in \{0,1\}\Big| \sum_{a \in \mathcal{A}}x_{akt}\leq 1,\sum_{k=1}^{K_1}x_{akt} \leq K_1, \forall a \in \mathcal{A}_{j\geq t}\right\}, \label{opt1:d}\\
 \boldsymbol{\tau}_t\in \mathcal{Y}=\\\notag\left\{\tau_{at} \in \left[0,\tau\right]\Big| \sum_{a \in \mathcal{G}_{t,2}} \tau_{at}+|\mathcal{G}_{t,2}|\tau' \leq \tau, \forall a \in \mathcal{G}_{t,2}\right\}, \label{opt1:e}\\
 \boldsymbol{x}_t,\boldsymbol{\tau}_t\in \mathcal{Z}=\left\{\boldsymbol{x}_t\in \mathcal{X},\boldsymbol{\tau}_t\in \mathcal{Y}\Big|\sum_{k=1}^{K_1}x_{akt}\tau_{at}=0\right\},\,\label{opt1:f}
\end{IEEEeqnarray}
\end{subequations}
where $\lambda_{t,1}$ and $\lambda_{t,2}$ denote, respectively, the number of satisfied UAs scheduled at $\mu$W and mmW bands at slot $t$.
Given a decision policy $\pi_t$, $\lambda_{t,2}$ is a random variable that depends on $\zeta_{at}$
at the mmW band. In fact, the expectation in \ref{opt1:a} is taken over $\zeta_{at}$, for all $a \in \mathcal{G}_{t,2}$. However, $\lambda_{t,1}$ is deterministic, if the slot duration $\tau$ is smaller than the $\mu$W channel coherence time.
\begin{figure}
\centering
\centerline{\includegraphics[width=\columnwidth]{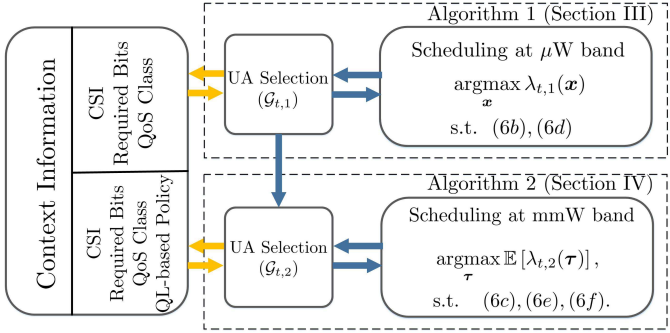}}\vspace{-0.2cm}
\caption{\small The structure of the context-aware scheduler.}\vspace{-1.5em}
\label{organ}
\end{figure}

The problem \eqref{opt1:a}-\eqref{opt1:f} falls into a class of optimization problems, referred to as \textit{minimum unsatisfied relations (Min-UR)}, which are known to be NP-hard \cite{Amaldi98}. 
Although linear systems with equality or inequality constraints can be solved in polynomial time, using an adequate linear programming method, least mean squared methods are not appropriate for infeasible systems when the objective is to maximize satisfied relations \cite{Amaldi98}. 

With this in mind, we propose a two-stage solution that solves \eqref{opt1:a}-\eqref{opt1:e} in polynomial time, as illustrated in Fig. \ref{organ}. The scheduling at $\mu$W band is considered first in order to reliably  schedule as many UAs as possible with small required loads over the $\mu$W band. \textcolor{black}{The motivation for serving UAs first at $\mu$W band is due to the fact that transmissions at $\mu$W frequencies are robust against blockage. Unlike $\mu$W frequencies, mmW communication is highly susceptible to blockage and, thus, scheduling  UAs only at the mmW band can potentially cause  outage for delay-intolerant UAs.} To this end, a hierarchy scheme is proposed based on the UAs' QoS class, CSI, and the required loads. Moreover, the UAs selection and scheduling are jointly done at SBS using an iterative algorithm. Then, for the remaining UAs that were not scheduled at the $\mu$W band, we propose a joint UA selection criterion  and scheduling algorithm that introduces a hierarchy to the UAs, based on the context information, and maximizes the number of satisfied UAs.


\vspace{0em}
\section{Context-Aware UA Selection and Resource Allocation at $\mu$W band}\label{Sec:III}
Before scheduling at mmW band, the goal of the scheduler is to first find an allocation $\boldsymbol{x}^*_t$ at each slot $t$ over $\mu$W band that satisfies
\begin{subequations}
\begin{IEEEeqnarray}{rCl}
&&\argmax_{\boldsymbol{x}^*_t} \lambda_{t,1}(\boldsymbol{x}_t^*),\label{opt2:a}\\
\text{s.t.}\,\,\,\,\,
&& \eqref{opt1:b}, \eqref{opt1:d}.\label{opt2:b}
\end{IEEEeqnarray}
\end{subequations}
The downlink scheduling problem in \eqref{opt2:a}-\eqref{opt2:b} is an inconsistent combinatorial problem of matching users to resources which does not admit a closed-form solution and has an exponential complexity\cite{4036195}. Hence, the solution of \eqref{opt2:a}-\eqref{opt2:b} depends on which UAs are chosen to be scheduled at $\mu$W band, i.e., the set $\mathcal{G}_{t,1}$. To this end, we introduce a hierarchy for UA selection by grouping the different UAs into the following sets: 
\begin{align}\label{uasel1}
\mathcal{G}^{(1)}_{t,1} &= \{a \in \mathcal{A}_j | j=t, b_a^{\textrm{req}}(t)>0 \},\\
\mathcal{G}^{(2)}_{t,1} &= \{a \in \mathcal{A}_j | j>t, b_a^{\textrm{req}}(t)>0 \}.
\end{align}

In fact, the UAs in $\mathcal{G}^{(1)}_{t,1}$ have higher priority than $\mathcal{G}^{(2)}_{t,1}$, since they must be served during the current time slot, otherwise, they will experience an outage. In addition, for UAs of the same set, the UA that satisfies the following has the highest priority:
\begin{align}\label{uasel2}
a^* = \argmin_{a} \frac{b_a^{\textrm{req}}(t)}{\sum_{k \in \mathcal{K}_1}R_a(k,t)},
\end{align}
where \eqref{uasel2} selects UA $a^*$ that minimizes the ratio of the required load to the achievable rate. To ensure that the constraints set for the selected  UAs $a \in \mathcal{G}_{t,1}$ is feasible, i.e. $\lambda_{t,1}(\boldsymbol{x})=|\mathcal{G}_{t,1}|$, the UA selection has to be done jointly while solving \eqref{opt2:a}-\eqref{opt2:b}. Following, we propose a framework that solves \eqref{opt2:a}-\eqref{opt2:b} for a given $\mathcal{G}_{t,1}$.   
\vspace{-0em}
\subsection{Scheduling as a Matching Game: Preliminaries}
For a selected set of UAs at $\mu$W band, $\mathcal{G}_{t,1}$, we propose a novel resource allocation scheme at $\mu$W band based on \emph{matching theory}, a mathematical framework that provides polynomial time solutions for complex combinatorial problems such as the one in \eqref{opt2:a}-\eqref{opt2:b} \cite{Roth92,eduard11,6853635,7155571}. A matching game is defined as a two-sided assignment problem between two disjoint sets of players in which the players of each set are interested to be matched to the players of the other set, according to \textit{preference relations}. At each time slot $t$ of our scheduling problem, $\mathcal{K}_1$ and $\mathcal{G}_{t,1}$ are the two sets of players. A preference relation $\succ$ is defined as a complete, reflexive, and transitive binary relation between the elements of a given set. Here, we let $\succ_a$ be the preference relation of UA $a$ and denote $k\succ_a k'$, if player $a$ prefers RB $k$ over RB $k'$. Similarly, we use $\succ_k$ to denote the preference relation of RB $k \in \mathcal{K}_1$.

In the proposed scheduling problem, the preference relations of UAs depend on both the rate and the QoS constraint which will be captured via well-designed, individual utility functions for UAs and SBS resources, as defined later in this section. 
\vspace{-0em}
\subsection{Scheduling at $\mu$W band as a Matching Game}
Each scheduling decision $\pi_{t,1}$ determines the allocation of RBs to UAs during time slot $t$ over the $\mu$W band. Thus, the scheduling problem at $\mu$W frequency band can be defined as a \textit{one-to-many matching game}:
\begin{definition}
Given two disjoint finite sets of players $\mathcal{G}_{t,1}$ and $\mathcal{K}_1$, the scheduling decision at time slot $t$, $\pi_{t,1}$, can be defined as a \textit{matching relation}, $\pi_{t,1}:\mathcal{G}_{t,1}  \rightarrow \mathcal{K}_1$ that satisfies 1) $\forall a \in \mathcal{G}_{t,1}, \pi_{t,1}(a) \subseteq \mathcal{K}_1$, 2) $\forall k \in \mathcal{K}_1, \pi_{t,1}(k) \in \mathcal{G}_{t,1}$, and 3) $\pi_{t,1}(k)=a$, if and only if  $k \in \pi_{t,1}(a)$.
\end{definition}    
In fact, $\pi_{t,1}(k)=a$ implies that $x_{akt}=1$, otherwise $x_{akt}=0$. Therefore, $\pi_{t,1}$ is indeed the scheduling decision that determines the allocation at $\mu$W band. One can easily see from the above definition that the proposed matching game inherently satisfies the constraint \eqref{opt1:d}. Next, we need to define suitable utility functions to determine the preference profiles of UAs and RBs. Given matching $\pi_{t,1}$, we define the utility of UA $a$ for $k \in \mathcal{K}_1$ at time slot $t$ as:
\begin{align}\label{utility1}
\Psi_a(k,t;\pi_{t,1})=
&\begin{cases}
0             &\text{if}\,\, \sum\limits_{k' \in \pi_{t,1}(a)}\!\!\!R_a(k',t)\tau\geq b_{a}^{\textrm{req}}(t),\\
R_a(k,t) &\text{otherwise.} 
\end{cases}
\end{align}

The utility of $\mu$W RBs $k \in \mathcal{K}_1$ for UA $a \in \mathcal{G}_{t,1}$ is simply the rate
\begin{align}\label{utility2}
\Phi_k(a,t)=R_a(k,t).
\end{align}
Using these utilities, the preference relations of UAs and RBs at a given time slot $t$ will be
\begin{IEEEeqnarray}{rCl}\label{prefer}
k \succ_a k' &&\Leftrightarrow \Psi_a(k,t;\pi_{t,1}) \geq \Psi_a(k',t;\pi_{t,1}) \label{prefer1}\\
a \succ_k a' &&\Leftrightarrow \Phi_k(a,t) \geq \Phi_k(a',t) \label{prefer2}, 
\end{IEEEeqnarray}
for $\forall a, a' \in \mathcal{G}_{t,1}$, and $\forall k,k' \in \mathcal{K}_1$. Given this framework, we propose a joint UA selection and matching-based scheduling algorithm that maximizes $\lambda_{t,1}$. 
\subsection{Proposed Context-aware Scheduling Algorithm at $\mu$W Band}
To solve the proposed game and find a suitable outcome, we use the concept of two-sided \textit{stable matching} between UAs and RBs, defined as follows \cite{Roth92}:

\begin{definition}
A pair $(a,k) \notin \pi_{t,1}$ is said to be a \textit{blocking pair} of the matching $\pi_{t,1}$, if and only if $a \succ_{k} \pi_{t,1}(k)$ and $k \succ_a \pi_{t,1}(a)$.
Matching $\pi_{t,1}$ is \textit{stable}, if there is no blocking pair.
\end{definition}
A stable scheduling decision, $\pi_{t,1}$, ensures fairness for the UAs. That is, if a UA $a$ envies the allocation of another UA $a'$, then $a'$ must be preferred by the RB $\pi_{t,1}(a')$ to $a$, i.e., the envy of UA $a$ is not justified.  
\begin{algorithm}[!t]
\footnotesize
\caption{Context-Aware UA Selection and Resource Allocation Algorithm at $\mu$W Band}\label{algo:1}
\textbf{Inputs:}\,\,$\mathcal{G}^{(1)}_{t,1}$,\,\,$\mathcal{G}^{(2)}_{t,1}$,\,\,$\boldsymbol{b}^{\textrm{req}}(t)$,\,\,$R_{a}(k,t)$.\\
\textbf{Outputs:}\,\, $\boldsymbol{x}$; $\mathcal{G}_{t,1}$.\\
\textit{Initialize}: $\mathcal{G}_{t,1}=\emptyset$,
\begin{algorithmic}[1]
\State  $\mathcal{G'}_{t,1}=\mathcal{G}^{(1)}_{t,1}$, $\mathcal{K}_a = \mathcal{K}_1, \forall a \in \mathcal{G'}_{t,1}$.
\State Add UA $a^* \in \mathcal{G'}_{t,1}$ with smallest $b_{a}^{\textrm{req}}(t)/\sum_{k \in \mathcal{K}_1}R_a(k,t)$ to $\mathcal{G}_{t,1}$ and remove it from $\mathcal{G'}_{t,1}$.
\State Update the preference ordering of UAs $a \in \mathcal{G}_{t,1}$ and RBs $k \in \mathcal{K}_1$, using \eqref{utility1} and \eqref{utility2}.
\State Using $\succ_a$, a UA $a \in \mathcal{G}_{t,1}$ is tentatively assigned to its most preferred RB in $\mathcal{K}_a$.
\State From the tentative list of UA applicants plus $\pi_{t,1}(k)$ for RB $k$, only the most preferred UA, based on $\succ_k$, is assigned to $k$. Next, $k$ is removed from the applicants' $\mathcal{K}_a$ sets.
\State Each UA $a$ updates $b_{a}^{\textrm{req}}(t)$ and $\succ_a$ based on \ref{prefer1}.
\Repeat \,\,\,Steps $3$ to $6$ \Until{$\mathds{1}_a(\boldsymbol{x})=1$, or $\mathcal{K}_a = \emptyset, \forall a \in \mathcal{G}_{t,1}$.} 
\If{$\exists a\in \mathcal{G}_{t,1},  \mathds{1}_a(\boldsymbol{x})\neq1$}
\State Remove $a^*$ from $\mathcal{G}_{t,1}$ and go to Step $3$.
\EndIf
\If{$\exists k, \sum_{a \in \mathcal{G}^{(1)}_{t,1} \cup \mathcal{G}^{(2)}_{t,1}}x_{akt}=0$}
let $\mathcal{G'}_{t,1} = \mathcal{G}_{t,1}^{(2)}$ and go to Step 2.\EndIf 
\end{algorithmic}\label{Algorithm1}
\end{algorithm}\setlength{\textfloatsep}{.2\baselineskip}
For conventional matching problems, the popular \emph{deferred acceptance (DA)} algorithm is normally used to find a stable matching \cite{Roth92,eduard11,7105641}. 
However, DA cannot be applied directly to our problem because it assumes that the quota for each UA is fixed. The quota is defined as the maximum number of RBs that a UA can be matched to. In our problem, however, quotas cannot be predetermined, since the number of RBs needed to satisfy the QoS constraint of a UA in \eqref{opt2:b} depends on the channel quality at each RB, as well as the context information. In fact, the adopted utility functions in \eqref{utility1} depend on the current state of the matching. Due to the dependency of UAs' preferences to the state of the matching, i.e. $x_{akt}$ variables, the proposed game can be classified as a \textit{matching game with externalities} \cite{7105641}. For matching games with externalities, DA may not converge to a two-sided stable matching. Therefore, a new matching algorithm must be found to solve the problem.

To this end, we propose a novel context-aware scheduling algorithm shown in Algorithm \ref{algo:1}. Algorithm \ref{algo:1} first allocates the RBs to the UAs in $\mathcal{G}^{(1)}_{t,1}$. At every iteration, each UA $a^*$ given by \eqref{uasel2} is added to the set $\mathcal{G}_{t,1}$ of the matching game. In Steps $4$ to $10$, the algorithm assigns RBs $k \in \mathcal{K}_1$ to UAs $a \in \mathcal{G}_{t,1}$ as follows. Each UA $a \in \mathcal{G}_{t,1}$ is tentatively assigned to its most preferred RB $k \in \mathcal{K}_a$. Next, from the tentative list of candidate UAs as well as current assignment  $\pi_{t,1}(k)$, the scheduler allocates RB $k$ only to the most preferred UA, based on $\succ_k$. The RB $k$ is removed from the set $\mathcal{K}_a$  corresponding to each candidate UA $a \in \mathcal{G}_{t,1}$. Based on the allocated RBs, the UAs update $b_a^{\textrm{req}}(t)$ and $\succ_a$. This process is repeated until the rate constraints for UAs are satisfied, $\mathds{1}_a(\boldsymbol{x})=1$, or $\mathcal{K}_a=\emptyset$ for UAs $a \in \mathcal{G}_{t,1}$. Then, if some of the RBs are left unallocated, the algorithm follows Steps $2$ to $14$ to add UAs from $\mathcal{G}_{t,1}^{(2)}$ to $\mathcal{G}_{t,1}$.
\begin{theorem}
The proposed Algorithm \ref{Algorithm1} is guaranteed to yield a two-sided stable matching between UAs and $\mu$W RBs.
\end{theorem}
\begin{proof}
The convergence of the Algorithm \ref{Algorithm1} at each slot is guaranteed, since a UA never applies for a certain RB twice. Hence, at the worst case scenario, all UAs will apply for all RBs once, which yields $\mathcal{K}_{a}=\emptyset, \forall a \in \mathcal{A}$. Next, we show that, once the algorithm converges, the resulting matching between UAs and RBs is two-sided stable. Assume that there exists a pair $(a,k) \notin \pi_{t,1}$ that blocks $\pi_{t,1}$. Since the algorithm has converged, we can conclude that at least one of the following cases is true about $a$: $\mathds{1}_a(\boldsymbol{x})=1$ , or $\mathcal{K}_{a}=\emptyset$.

The first case, $\mathds{1}_a(\boldsymbol{x})=1$ implies that $a$ does not need to add more RBs to $\pi_{t,1}(a)$. In addition, $a$ would not replace any of $k' \in \pi_t(a)$ with $k$, since $k' \succ_{a} k$. Otherwise, $a$ would apply earlier for $k$. If $a$ has applied for $k$ and got rejected, this means $\pi_{t,1}(k) \succ_k a$, which contradicts $(a,k)$ to be a blocking pair. Analogous to the first case, $\mathcal{K}_{a}=\emptyset$ implies that $a$ has got rejected by $k$, which means $\pi_{t,1}(k) \succ_k a$ and $(a,k)$ cannot be a blocking pair. This proves the theorem.
\end{proof}
Given $\mathcal{G}_{t,1}$ by Algorithm \ref{Algorithm1} at $\mu$W band, the scheduling problem at slot $t$ is now reduced to choosing a subset of unscheduled UAs and allocate mmW resources to them such that the number of satisfied UAs is maximized. 
\section{Context-Aware UA Selection and Resource Allocation at mmW Band}\label{Sec:IV}
 We let $ \mathcal{G'}_{t,2}\!=\!\{ a \in \mathcal{A}_{j\geq t}| a \notin \mathcal{G}_{t,1}, b_{a}^{\textrm{req}}(t)\!>0\}$ be the set of UAs that have not been scheduled  over the $\mu$W band. Here, the scheduling problem over the mmW band at slot $t$ can be formulated as a stochastic min-UR problem as follows:
\begin{subequations}
\begin{IEEEeqnarray}{rCl}
&&\argmax_{\boldsymbol{\tau}_t} \mathbb{E}\left[\lambda_{t,2}(\boldsymbol{\tau}_t)\right],\label{opt3:a}\\
\text{s.t.}\,\,\,\,\,
&& \eqref{opt1:c}, \eqref{opt1:e}, \eqref{opt1:f}.\label{opt3:b}
\end{IEEEeqnarray}
\end{subequations}

Here, we note that $\zeta_{at}$ in \eqref{opt1:c} is a Bernoulli random variable with success probability $\rho_a$. Hence, for any allocation $\boldsymbol{\tau}_t$, the number of satisfied constraints $\lambda_{t,2}$ is a random variable. Although the exact distribution of $\lambda_{t,2}$ may not be found for a general infeasible problem as \eqref{opt3:a}-\eqref{opt3:b}, we can approximate the distribution of outage ratio at slot $t$, $P_{\textrm{out,t}} = 1-\left[(\lambda_{t,1}+\lambda_{t,2})/A_{j=t}\right]$, as follows: 
\begin{proposition}
Let $\boldsymbol{\tau}_t$ be a feasible solution for the subset of constraints in \eqref{opt3:b} associated with UAs $a \in \mathcal{G}_{t,2} \subseteq \mathcal{G'}_{t,2}$. Given $\lambda_{t,1}$ and $0 \leq P_{th}<1-\frac{\lambda_{t,1}}{A_{j=t}} $, where $P_{th}$ is an outage threshold, the CDF of the outage ratio at slot $t$, $F_{P_{out,t}}(P_{th})$ can be approximated by,
\begin{align}\label{prop1-1}
F_{P_{out,t}}(P_{th}) \approx 1-\frac{\Gamma\left(\lfloor (1-P_{th})A_t-\lambda_{t,1}+1 \rfloor,\lambda_{ave}\right)}{\lfloor (1-P_{th})A_t-\lambda_{t,1} \rfloor !},
\end{align}
where $\lfloor . \rfloor$ is the floor function, $\Gamma(.,.)$ is the incomplete gamma function, and 
\begin{align}\label{exp2}
\lambda_{ave}=\mathbb{E}\left[\lambda_{t,2}(\boldsymbol{\tau}_t)\right]=\sum_{a \in \mathcal{G}_{t,2}}\rho_{a}.
\end{align}
\end{proposition}
\begin{proof}
We can write $\lambda_{t,2}$ as the sum of Bernoulli random variables $\zeta_{at}$, i.e., $\lambda_{t,2}(\boldsymbol{\tau})=\sum_{a \in \mathcal{G}_{t,2}}\zeta_{at}$. Hence, using Le Cam's theorem, the distribution of $\lambda_{t,2}$ follows Poisson distribution, i.e.,
\begin{align}\label{exp1}
\mathbb{P}\left[\lambda_{t,2}(\boldsymbol{\tau}_t)=k\right]\approx \frac{\lambda_{ave}^ke^{-\lambda_{ave}}}{k!},
\end{align}
where $\mathbb{E}\left[\lambda_{t,2}(\boldsymbol{\tau}_t)\right]$ is the sum of expected values of $\zeta_{at}$ for selected UAs in $\mathcal{G}_{t,2}$ as given by \eqref{exp2}. Moreover, the approximation error is bounded by
\begin{align}\label{LeCam}
\sum_{k=0}^{\infty}\Big|\mathbb{P}\left[\lambda_t=k\right]-\frac{\lambda_{ave}^ke^{-\lambda_{ave}}}{k!}\Big|<2\sum_{a \in \mathcal{G}_{t,2}}\rho_{a}^2,
\end{align}
where $\lambda_t = \lambda_{t,1}+\lambda_{t,2}$. Next, 
\begin{align}
\!\!\!F_{P_{out,t}}(P_{th})&=\mathbb{P}\left(P_{out,t}\leq P_{th}\right)\notag\\
&=\mathbb{P}\left(1-\frac{\lambda_{t,1}+\lambda_{t,2}}{A_t}\leq P_{th}\right)\\
&= 1- \mathbb{P}\left(\lambda_{t,2}\leq\lfloor(1-P_{th})A_t-\lambda_{t,1}\rfloor\right)\\
&\approx 1\!-\!\frac{\Gamma\left(\lfloor (1-P_{th})A_t-\lambda_{t,1}+1 \rfloor,\lambda_{ave}\right)}{\lfloor (1-P_{th})A_t-\lambda_{t,1} \rfloor !},\label{LeCam2-22}
\end{align}
where \eqref{LeCam2-22} follows the CDF of the Poisson distribution. 
\end{proof}
From \eqref{exp2} and \eqref{LeCam}, we can see that the objective function increases as UAs with higher $\rho_{a}$ are satisfied, however, the approximation of the distribution becomes less accurate.   

We note that if LoS probabilities $\rho_a$ are known by the SBS, the proposed scheduling problem over the mmW band becomes equivalent to a 0-1 stochastic Knapsack optimization problem \cite{Dean05}. However, in practice, the explicit values of $\rho_a$ may not be available at the SBS. In Section \ref{sec:QL}, we will introduce a learning approach using which the SBS can determine if $\rho_a \geq \rho_{th}$, where $\rho_{th}$ is a constant value. By learning which UAs satisfy $\rho_a \geq \rho_{th}$, the SBS assigns priority to the UAs that are more likely to be at a LoS link from the SBS. This information along with the QoS classes of UAs will allow the scheduler to group UAs into the following non-overlapping subsets:
\begin{eqnarray}
\mathcal{G}^{(1)}_{t,2} &= \{a \in \mathcal{A}_{j=t} \cap \mathcal{G'}_{t,2}| \rho_a \geq \rho_{th} \},\\\label{QL1_a}
\mathcal{G}^{(2)}_{t,2} &= \{a \in \mathcal{A}_{j=t} \cap \mathcal{G'}_{t,2}| \rho_a < \rho_{th}   \},\\\label{QL1_b}
\mathcal{G}^{(3)}_{t,2} &= \{a \in \mathcal{A}_{j>t}\cap \mathcal{G'}_{t,2}|\rho_a \geq \rho_{th}\},\\\label{QL1_c}
\mathcal{G}^{(4)}_{t,2} &= \{a \in \mathcal{A}_{j>t}\cap \mathcal{G'}_{t,2}|\rho_a < \rho_{th}\}.\label{QL1_d}
\end{eqnarray}

In fact, the SBS will adopt a greedy approach that assigns priority to sets $\mathcal{G}^{(i)}_{t,2}$ with $i=1$ as highest and $i=4$ as lowest priority. That is due to the fact that UAs in $\mathcal{G}^{(1)}_{t,2}$ cannot tolerate further delays. Moreover, they belong to UEs with high possibility of LoS access to SBS. In addition, UAs in $\mathcal{G}^{(2)}_{t,2}$ are in second priority, since they cannot tolerate more delay, while having a low $\rho_a$. Moreover, UAs $\mathcal{G}^{(3)}_{t,2}$ are assigned to a third priority, since they can tolerate more delays and have high probability to be at LoS mmW link with SBS. The least priority is assigned to UAs in $\mathcal{G}^{(4)}_{t,2}$ as they can tolerate further delays, while having low $\rho_a$. 

Furthermore, for the UAs of the same set, the highest priority is given to a UA $a^*$ that satisfies: \vspace{0em}
\begin{align}\label{uasel3}
a^* = \argmin_{a} \frac{b_a^{\textrm{req}}(t)}{R_a(t)}.
\end{align}
In other words, the SBS selects the UA that requires the least time resource to be satisfied. Similar to $\mu$W band scheduling, the SBS must ensure that the constraints set for selected  UAs $a \in \mathcal{G}_{t,2}$ is feasible. Therefore, the UA selection has to be done jointly while solving \eqref{opt3:a}-\eqref{opt3:b}. Next, we propose a joint UA selection and scheduling algorithm at mmW band.
\subsection{Proposed Context-aware Scheduling Algorithm over the mmW Band}
\begin{algorithm}[!t]
\footnotesize
\caption{Context-Aware UA Selection and Resource Allocation Algorithm at mmW Band}\label{euclid}
\textbf{Inputs:}\,\,$\mathcal{G}^{(i)}_{t,2},i=1,...,4$,\,\,$\boldsymbol{b}^{\textrm{req}}(t)$,\,\,$R_{a}(t)$.\\
\textbf{Output:}\,\,$\boldsymbol{\tau};\mathcal{G}_{t,2}$.
\begin{algorithmic}[1]
\State \textit{Initialize:} $\mathcal{G}_{t,2}=\emptyset$.
\For{$i=1; i \leq 4; i++$}
\For{$j=1:|\mathcal{G}^{(i)}_{t,2}|$}
\State Find UA $a^* \in \mathcal{G}^{(i)}_{t,2}$ from \eqref{uasel3}, set $\tau_{a^*,t} = b_{a^*}^{\textrm{req}}(t)/R_{a^*}(t)$ and add $a^*$ to $\mathcal{G}_{t,2}$.
\If{(\ref{opt1:e}) is not satisfied} 
\State Remove $a^*$ from $\mathcal{G}_{t,2}$. Break. 
\EndIf
\EndFor
\EndFor 
\end{algorithmic}\label{Algorithm2}
\end{algorithm}\setlength{\textfloatsep}{.2\baselineskip}
Over the mmW band, the objective is to serve as many UAs as possible in order to offload more traffic from the $\mu$W band, subject to UAs delay constraints. With this in mind, we propose Algorithm \ref{Algorithm2} to solve \eqref{opt3:a}-\eqref{opt3:b}. The algorithm follows the priority criterion introduced in \eqref{QL1_a}-\eqref{QL1_d}. Starting with the set $\mathcal{G}^{(1)}_{t,2}$, the scheduling process is a 0-1 Knapsack problem composed of $|\mathcal{G}^{(1)}_{t,2}|$ items all with the same benefit $\rho_{th}$ and weights equal to the required time $\tau_{a,t}=\frac{b_a^{\textrm{req}}(t)}{R_a(t)}$. This problem can be simply solved by sorting the required time in increasing order and adding UAs one by one to the set $\mathcal{G}_{t,2}$. The algorithm follows the process for the remaining sets and converges, once the entire mmW time slot duration is allocated and no additional time is available for more UAs. \textcolor{black}{From Algorithms \ref{algo:1} and \ref{Algorithm2}, we observe that resource allocation at any time slot affects the scheduling at both mmW and $\mu$W bands for the subsequent time slots.  Therefore, the proposed UA selection and scheduling schemes at one frequency band are not independent of those at the other frequency band, thus requiring \emph{joint scheduling} for the dual-mode system.}

The above solution requires the SBS to determine for which UAs the condition $\rho_a \geq \rho_{th}$ is satisfied. Next, we introduce a learning scheme that enables the UEs to obtain this information by monitoring successful LoS transmissions from the SBS over time and send it to the SBS. Clearly, $\rho_a$ is the same for the UAs that run at an arbitrary UE, since they experience the same wireless channel. \vspace{0em}
\subsection{Q-learning Model to Evaluate the LoS Probability}\label{sec:QL}
In a real-world cellular network, the UEs will be surrounded by many objects and, thus, the SBS may never know in advance whether an LoS mmW link will be available or not. Therefore, scheduling UAs of UEs that are experiencing a high blockage probability not only wastes network resources, it may drastically degrade QoS for delay intolerant UAs.

In practice, $\rho_a$ depends on many parameters such as the distance between the UE and the SBS, or blockage by human or other surrounding objects. Although finding a closed-form relation of $\rho_a$ with these parameters may not be feasible in general, the UEs can learn whether they have a high LoS probability based on transmissions from the SBS over time. The UEs will then update and send this information to the SBS at each time slot. Clearly, a simple averaging over time would not work, since the environment is dynamic and $\rho_a$ may change over time. To this end, we propose a learning framework, based on Q-learning (QL) \cite{Sutton}, in order to determine UAs with $\rho_a \geq \rho_{th}$ without knowing the actual $\rho_{a}$ values. QL is a reinforcement learning algorithm that determines optimal policy without detailed modeling of the system environment \cite{Sutton,6503987}.
The proposed QL model is formally defined by the following key elements:
\begin{itemize}
\item \textit{Agents:} UEs $m \in \mathcal{M}$.
\item \textit{States:} Depending on whether a UA of a given UE is being scheduled over $\mu$W or mmW bands, there are three possible states for the UA: 1) UA is served by the SBS over a LoS mmW link ($S_1$), 2) UA is scheduled over mmW band, but no LoS link is possible ($S_2$), and 3) UA is scheduled over $\mu$W band ($S_3$).  
\item \textit{Action:} At any state, a UE can make a decision $d$ chosen from a set $\mathcal{D}=\{d_1,d_2\}$ where $d_1$ and $d_2$, respectively, stand for whether to schedule this user's UAs at the current frequency band or switch to the other frequency band.
\item \textit{State transition probability:} $T(S_i,d,S_j)$ denotes the probability  of transition from state $S_i$ to $S_j$ if decision $d \in \mathcal{D}$ is chosen by the UE. Hence, $T(S_i,d_2,S_3)=1$ for $i=1,2$, and $T(S_3,d_2,S_1)=1-T(S_3,d_2,S_2)=\rho_a$. In addition, $T(S_3,d_1,S_3)=1$ and $T(S_i,d_1,S_1)=1-T(S_i,d_1,S_2)=\rho_a$ for $i=1,2$. 
\item \textit{Reward:} The UE receives rewards $\boldsymbol{r}=\left[r_1,-r_2,r_3\right]$, respectively, for each of its UAs being at states $S_1$, $S_2$, and $S_3$, where $r_2>r_1>r_3>0$. The rewards are assumed the same for all UAs $a \in \mathcal{A}$. The reward values affect both the convergence and the policy. For instance, for large negative rewards, i.e., $r_2 \gg r_3$, the optimal policy for the UA is to choose $\mu$W, even for large $\rho_a$ values. The long-term reward for choosing mmW band by UA $a \in \mathcal{A}$ is $r_1 \rho_a -(1-\rho_a)r_2$. Therefore, we can set $\boldsymbol{r}$ such that only for $\rho_a \geq \rho_{th}$, mmW band be preferred by UA $a$. That is, $r_1 \rho_a -(1-\rho_a)r_2 \geq r_3$ which implies
\begin{align}\label{QL_r}
\rho_{th}=\frac{r_3+r_2}{r_1+r_2}, \,\,\, r_2>r_1>r_3>0.
\end{align}
\end{itemize} 
\begin{figure}
\centering
\centerline{\includegraphics[width=\columnwidth]{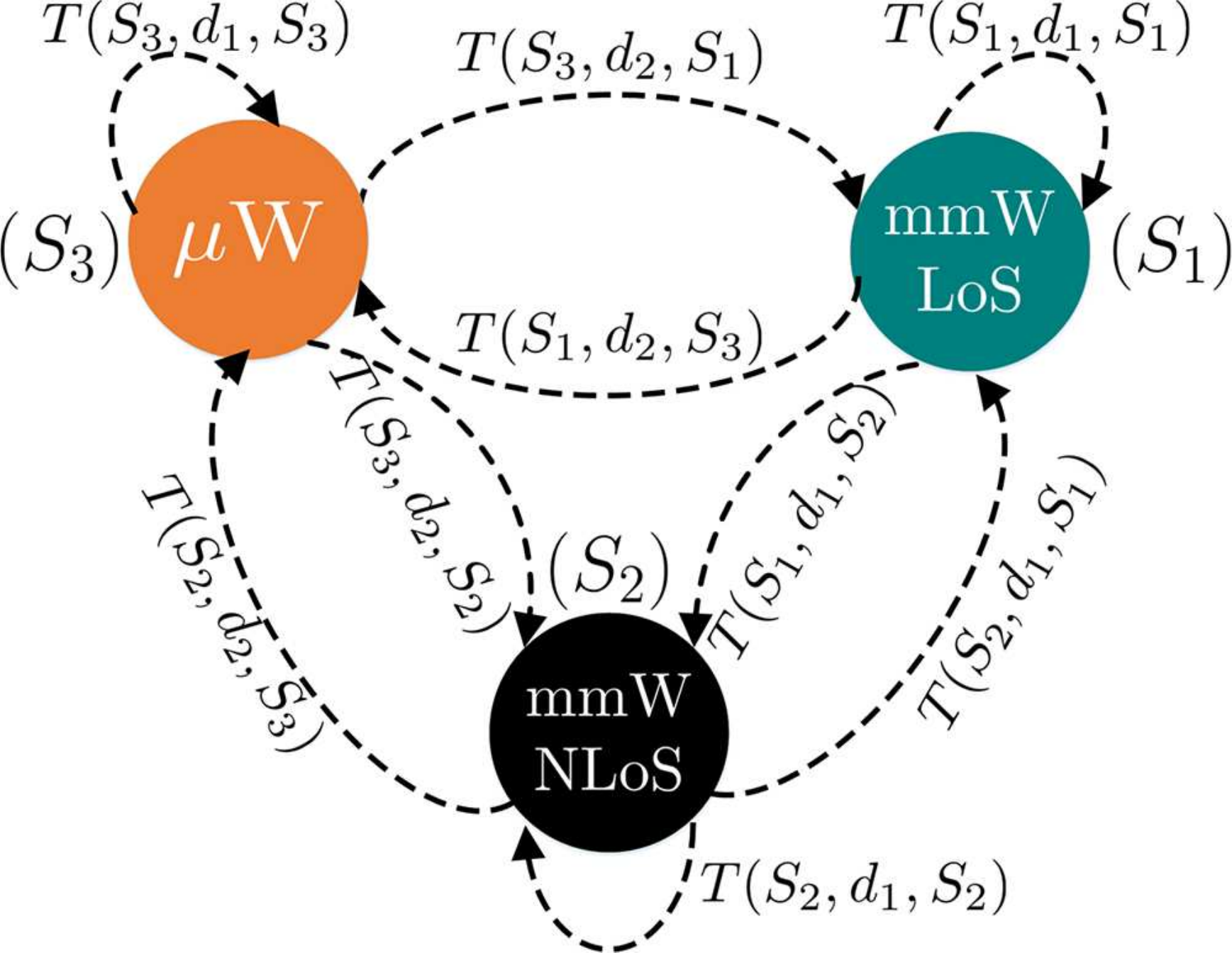}}\vspace{-0.1cm}
\caption{\small QL model with state transition probabilities.}\vspace{-0cm}
\label{sim5}
\end{figure}
At any time slot, each UA that is selected for scheduling will explore one of the three states. Consequently, this UA's corresponding UE will achieve a reward associated with the current residing state. We note that the UEs do not have any prior knowledge about the transition probabilities in advance. However, QL provides a model-free approach that instead of estimating $\rho_a$, it allows UE to find the best decision while residing at each state. This is done by the notion of Q-values $Q(S,d)$ which represents the value of decision $d$ while being at state $S$. Starting from an initial Q-values, UA can find true values via an iterative process as follows:
\begin{align}\label{QL1}
Q(S,d)\leftarrow (1-\alpha)Q(S,d)+\alpha \left[\boldsymbol{r}(S')+\gamma \max_{d'} Q(S',d') \right],
\end{align}
where $\alpha$ and $\gamma$ are predetermined constants. It can be shown that updating the Q-table based on (\ref{QL1}) maximizes the long-term expected reward: $\bar{r}=\lim\limits_{T \rightarrow \infty}\frac{1}{T} \sum_{t=1}^{T} \boldsymbol{r}(S(t))$ \cite{Sutton}. 
Moreover, given the converged $Q$ values, the following sufficient condition can be used to find a subset of UAs with $\rho_{a} \geq \rho_{th}$:
\begin{align}\label{QL2}
&Q(S_i,d_2)\leq Q(S_i,d_1), i=1,2 \,\,\,\,\text{and},\notag\\
&Q(S_3,d_1)\leq Q(S_3,d_2) \Rightarrow \rho_a \geq \rho_{th} .
\end{align}

We note that if there is only one UE that is running only one UA, the criterion given by (\ref{QL1}) leads to making optimal decisions in terms of maximizing the expected reward. 
However, the multi-user resource allocation cannot be done only based on $\rho_a \geq \rho_{th}$ criterion. On the one hand, assigning mmW resources only to UAs with high $\rho_a$ and small required load will result low spectral efficiency. Moreover, UAs with small $\rho_a$ and large $b_{a}^{\textrm{req}}(t)$ will not meet their delay requirement, if they are scheduled over the $\mu$W frequency band. However, even with small $\rho_a$, it is still probable for these UAs to be served over a LoS mmW link. Therefore, multi-user scheduling enforces SBS to exploit per UA context information, i.e., required load per UA, delay constraint, as well as UEs-SBS channel diversity. \textcolor{black}{Here, it worth noting that exploiting side information such as the geographical location information of buildings could also facilitate learning the LoS probabilities \cite{6840343,7430349}.} 
\subsection{Complexity Analysis of the Proposed Two-stage Solution}
\begin{figure*}
\centering
\begin{subfigure}[b]{0.3\textwidth}
\includegraphics[width=\textwidth]{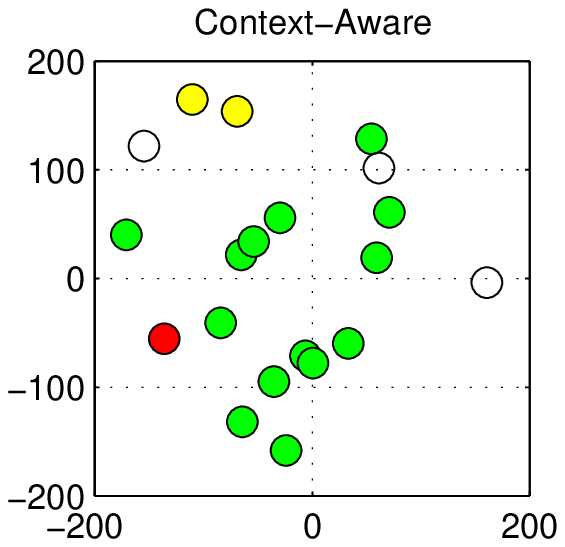}\vspace{-0.5cm}
\caption{Proposed approach}
\label{fig:context-aware}
\end{subfigure}
~\hspace{-.4cm} 
\begin{subfigure}[b]{0.3\textwidth}
\includegraphics[width=\textwidth]{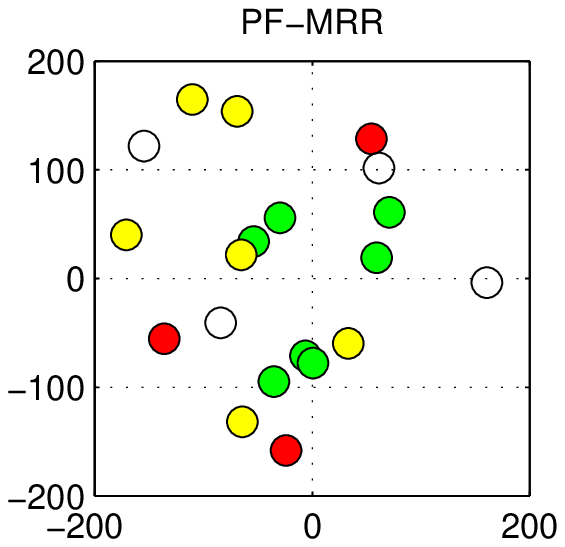}\vspace{-0.5cm}
\caption{PF-MRR}
\label{fig:PF-MRR}
\end{subfigure}
~\hspace{-.4cm} 
\begin{subfigure}[b]{0.3\textwidth}
\includegraphics[width=\textwidth]{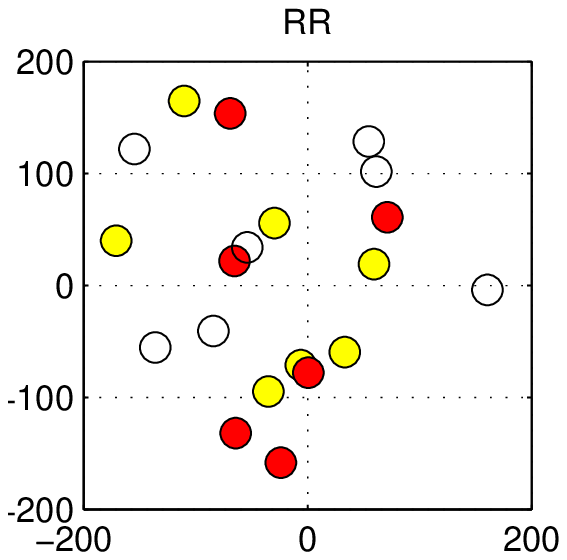}\vspace{-0.5cm}
\caption{RR}
\label{fig:RR}
\end{subfigure}\vspace{-.5em}
\caption{\small \textcolor{black}{Performance comparison between scheduling approaches for $M=20$ and $b_a=5$ Mbits. For the cell edge UEs, $\rho_a$ is sampled randomly from $\left[0,1\right]$ and for others $\rho_a=1$.}}\label{satisfymap-2}
\end{figure*}

With this in mind, we can make the following observation with regard to the proposed two-stage solution in Algorithms \ref{Algorithm1} and \ref{Algorithm2} for the original problem in \eqref{opt1:a}-\eqref{opt1:f}.
\begin{theorem}
The proposed long-term scheduling algorithm composed of Algorithm \ref{Algorithm1} and Algorithm \ref{Algorithm2} solves the problem \eqref{opt1:a}-\eqref{opt1:f} in polynomial time with respect to the number of UAs.
\end{theorem}\vspace{-1em}
\begin{proof}
First, we analyze the complexity of Algorithm \ref{Algorithm1}. For each slot $t$, let $A_{j\geq t}=|\mathcal{A}_{j\geq t}|$ be the number of UAs that can be selected to be scheduled at $\mu$W band. At most, the algorithm must find the solution for $A_{j\geq t}$ number of matchings. In addition, each matching has the complexity of $O(K_1)$, since in the worst case, each UA must be re-allocated to $K_1$ RBs by SBS. Hence, the complexity of Algorithm \ref{Algorithm1} at each slot $t$ is $O(K_1A_{j\geq t})$ and the total complexity from slot $t=1$ to $t=J$ is $O(K_1\sum_{j=1}^{J}jA_j)$.

Next, we analyze the complexity of Algorithm \ref{Algorithm2}. At each slot $t$, there are at most $A_{j\geq t}$ UAs to be scheduled at mmW ($\mathcal{G}_{t,1}=\emptyset$). Therefore, the Algorithm \ref{Algorithm2} must converge after $A_{j\geq t}$ resource allocations, where each allocation is a special case of the 0-1 Knapsack problem. Hence, the total complexity of the Algorithm \ref{Algorithm2} from slot $t=1$ to slot $t=J$ is $O(\sum_{j=1}^{J}jA_j)$.

From the above results, the overall complexity of the proposed long-term scheduling is $O\left((K_1+1)\sum_{j=1}^{J}jA_j\right)$.
\end{proof}
	\begin{table}[!t]
	\centering
	\caption{
		\vspace*{-0.1cm}Simulation parameters}\vspace*{-0.1cm}
	\begin{tabular}{|c|c|c|}
	\hline
	\bf{Notation} & \bf{Parameter} & \bf{Value} \\
	\hline
	$P_1,P_2$ & Transmit power & $30$ dBm\\
	\hline
	$(\Omega_1,\Omega_2)$ & Available Bandwidth & ($10$ MHz, $1$ GHz)\\
	\hline
	$\omega$ & Bandwidth per RB& $180$ KHz\\
	\hline
	\textcolor{black}{K-factor} & \textcolor{black}{Rician K-factor}& \textcolor{black}{$2.4$ \cite{7481410}}\\
	\hline
	($\xi_1, \xi_2$) & Standard deviation of mmW path loss& ($10, 5.2$) \cite{Ghosh14} \\
	\hline
	($\alpha_1, \alpha_2$) & Path loss exponent& ($3,2$) \cite{Ghosh14}\\
	\hline
	($\beta_1, \beta_2$) & Path loss at $1$ m& ($38, 70$) dB \\
	\hline
	$\psi$ & Antenna gain& $18$ dBi \\
	\hline
	$\tau$ & Time slot duration& $10$ ms \\
	\hline
	$\tau'$ & Beam-training overhead& $0.1$ ms \\
	\hline
	$N_0$ & Noise power spectral density& $-174$ dBm/Hz \\
	\hline
	$J$ & Number of QoS classes& $5$ \cite{Tang4543084}  \\
	\hline
	$\kappa$ & Number of UAs per UE& $3$  \\
	\hline
	$\boldsymbol{r}$ & Reward vector& $[3,-16,1]$  \\
	\hline
	\end{tabular}\label{tabsim}\vspace{.5em}
	\end{table}
	\vspace{-0em}
\section{Simulation Results}\label{Sec:V}
For simulations, we consider an area with diameter $d=200$ meters with the SBS located at the center \cite{Ghosh14}. UEs are distributed uniformly within this area with a minimum distance of $5$ meters from the SBS. Each UE has $\kappa$ UAs chosen randomly and uniformly from $J$ QoS classes. The main parameters are summarized in Table \ref{tabsim}. All statistical results are averaged over a large number of independent runs.
We compare the performance of the proposed context-aware algorithm with two well-known resource allocation approaches:
\begin{itemize}
\item Proportional Fair Scheduler with minimum rate requirement (PF-MRR): The PF scheduling for multi-carrier systems with minimum rate requirement is different than the conventional approach. In \cite{1683565}, a simple approach is proposed to implement PF-MRR which we modify to apply to the dual-mode system. At $\mu$W band, RB $k$ is assigned to the UA $a \in \mathcal{A}_t$ that satisfies 
\begin{align}
a = \argmax_{a \in \mathcal{A}_t}\frac{R_a(k,t)}{\bar{R}_{a}^{\textrm{rec}}(t)+R_{a}^{\textrm{req}}(t)},
\end{align}
where $\bar{R}_{a}^{\textrm{rec}}(t)$ is the achieved average rate up to time slot $t$, and $R_{a}^{\textrm{req}}(t)=b_{a}^{\textrm{req}}(t)/\tau$ is the required average rate at slot $t$ to meet the QoS constraint of UA $a$.
UAs $a \in \mathcal{A}_{t'\geq t}$ with unsatisfied rate requirement are scheduled at mmW band where $\tau_{a,t}=\frac{b_{a}^{\textrm{req}}(t)}{R_a(t)}$ is allocated to the UA $a=\argmax_{a} \frac{R_a(t)}{\bar{R}_{a}^{\textrm{rec}}(t)+R_{a}^{\textrm{req}}(t)}$, while $\sum_{a}\tau_{a,t}=\tau-|\mathcal{G}_{t,2}|\tau'$.
\item Round Robin Scheduler (RR): At $\mu$W band, the scheduler allocates equal number of RBs to each $a \in \mathcal{A}_t$. Unsatisfied UAs $a \in \mathcal{A}_{t'\geq t}$ are scheduled at mmW band with $\tau_{a,t}=\frac{\tau-|\mathcal{G}_{t,2}|\tau'}{|\mathcal{G}_{t,2}|}$.
\end{itemize}
\vspace{-0cm}
\subsection{Quality-of-Experience of the Users}


Fig. \ref{satisfymap-2} shows a snapshot of a given network realization in which specific UEs are represented by circles. Each UE is associated with $\kappa=3$ UAs, each having a required load of $b_a=5$ Mbits. \textcolor{black}{We note that for an arbitrary UA $a \in \mathcal{A}_j$, the required load $b_a$ (bits) can be translated into data rate $b_a/(j\tau)$. For example, $b_a=5$ Mbits for $a \in \mathcal{A}_5$ is equivalent to $100$ Mbits/s data rate.} The results from this figure show each user's satisfaction by indicating how many UAs per UE are satisfied. In Fig. \ref{satisfymap-2}, the colors red, yellow, and  green are used, respectively, to indicate one, two, and three satisfied UAs per UE. Moreover, circles with no color represent UEs with no serviced UA. \textcolor{black}{Clearly, in Fig. \ref{satisfymap-2}, we can see that the proposed approach significantly improves the overall system performance by providing service to more UEs, compared to both PF-MRR and RR schemes. In addition, we observe that the proposed context-aware approach outperforms PF-MRR and RR schemes by satisfying the QoS needs of more applications,  which naturally leads to a higher quality-of-experience per user.}

\vspace{-0cm}
\subsection{Outage Probability vs Number of UEs}
The overall outage probability, $P_{\textrm{out}}$, is defined as the ratio of the number of QoS violations over the total number of UAs which will be given by:
\begin{align}\label{lambda}
P_{\textrm{out}}(\pi)&= 1-\frac{1}{A}\left(\sum_{t=1}^{J} \lambda_{t,1}(\pi)-\sum_{t=1}^{J} \lambda_{t,2}(\pi)\right)\\\nonumber
&=1-\frac{1}{A}\sum_{t=1}^{J}\sum_{a \in \mathcal{A}_t}\mathds{1}(a ;\pi)=1-\frac{1}{A}|\mathcal{O}^{\pi}|,
\end{align}
Since $P_{\textrm{out}}$ is a random variable, we will study whether the proposed scheduling policy $\pi^* \in \boldsymbol{\Pi}$ guarantees
$\mathbb{P}(P_{\textrm{out}}(\pi^*) \geq P_{th}) \leq \epsilon,$
where $P_{th}$ is the maximum tolerable outage probability and $\epsilon$ is a pre-defined threshold. This can be written as $F_{P_{\textrm{out}}}(P_{th}) \geq 1-\epsilon$, where $F_{P_{\textrm{out}}}(.)$ is the cumulative distribution function (CDF) of $P_{\textrm{out}}$.

\begin{figure}
\centering
\centerline{\includegraphics[width=9cm]{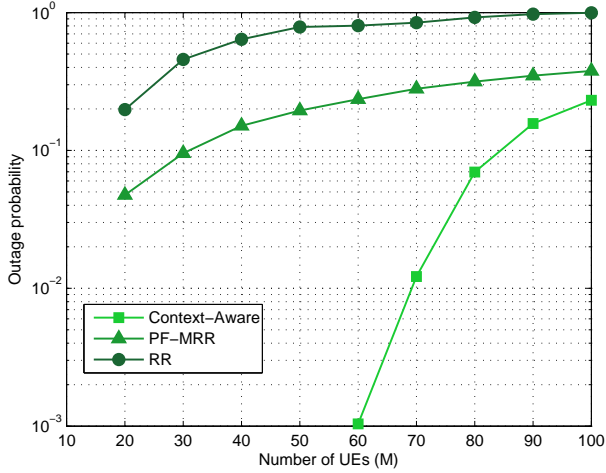}}\vspace{-0.2cm}
\caption{\small \textcolor{black}{Performance comparison between scheduling approaches versus the number of UEs, $M$. All users are at LoS. $b_a=1$ Mbits and $\rho_a=1$ for all $a \in \mathcal{A}$.}}\vspace{-0cm}
\label{fig3.4}
\end{figure}

Fig. \ref{fig3.4} shows the outage probability as the number of UEs varies, for the three considered approaches. Fig. \ref{fig3.4} shows that the outage probability increases as the number of UAs increases. In fact, the results show the number of UAs that can be satisfied for a given outage threshold. Clearly, the proposed algorithm outperforms the PF-MRR and RR scheduling approaches. \textcolor{black}{For example, for a $0.01$ outage probability, the proposed context-aware approach satisfies up to $210$ UAs, considering $\kappa=3$ UAs per UE. However, the baseline approaches fail to achieve this performance. In fact, the outage probability is always greater than $0.04$ for both the RR and the PF-MMR approaches over all network sizes. Finally, from Fig. \ref{fig3.4}, we can clearly see that  the proposed approach can always guarantee the QoS for up to $180$ UAs on average, which is three times greater than the number of satisfied UAs resulting from the PF-MRR and RR approaches.}

\vspace{-0em}
\subsection{Impact of Q-Learning}
\begin{figure}
\centering
\centerline{\includegraphics[width=9cm]{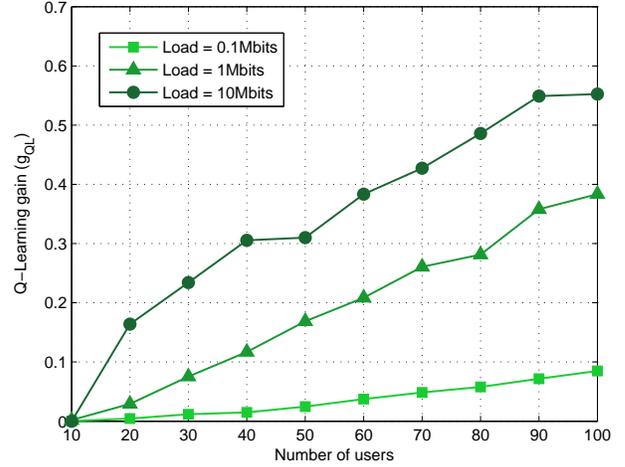}}\vspace{-0.2cm}
\caption{\small Gain of Q-Learning vs number of users for different load values. }\vspace{-0cm}
\label{fig3.6}
\end{figure}

Fig. \ref{fig3.6} shows the gain of the proposed QL approach. The QL gain is defined as the respective number of satisfied UAs with and without QL. The results presented in Fig. \ref{fig3.6} show that more gain is achievable as the number of UAs increases. This stems from the fact that, as the number of UAs increases, it becomes more probable that more number of UEs be at a LoS connection with the BS. Fig. \ref{fig3.6} shows that the QL-based information allows scheduling UAs with higher LoS probabilities. More interestingly, Fig. \ref{fig3.6} shows that the gain increases as the required load per UA increases. This is due to the fact that with more strict QoS constraints, it is become more important to allocate mmW resources only to the UAs with higher probability of LoS.

\begin{figure}
\centering
\centerline{\includegraphics[width=9cm]{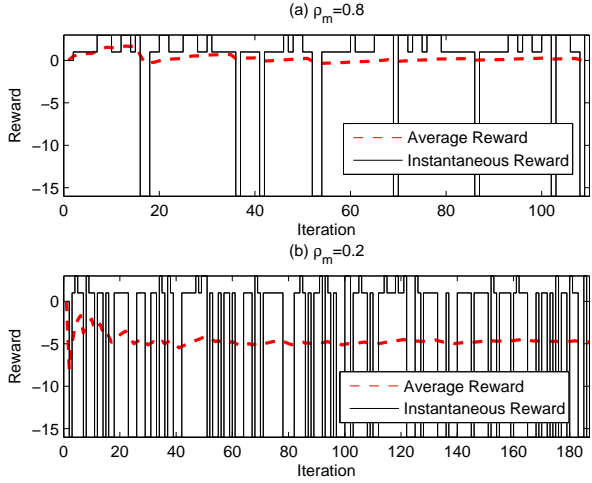}}\vspace{-0.2cm}
\caption{\small State exploration and convergence of Q-Learning }\vspace{-0.2cm}
\label{fig3.7}
\end{figure}

\begin{figure}
\centering
\centerline{\includegraphics[width=9cm]{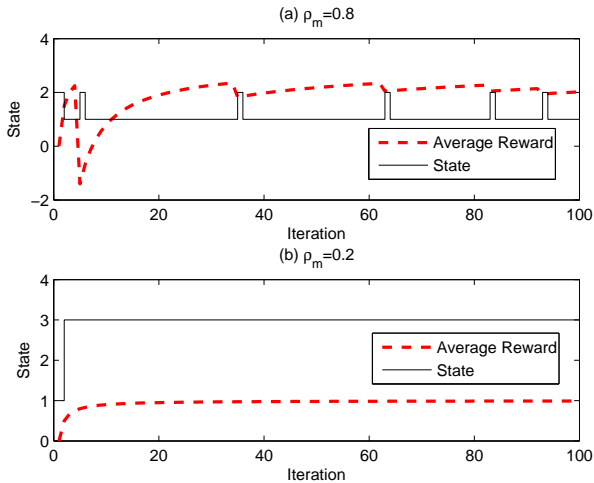}}\vspace{-0.2cm}
\caption{\small State transition and average reward resulted by the optimal policy.}\vspace{-0cm}
\label{fig3.8}
\end{figure}

Fig. \ref{fig3.7}.a and Fig. \ref{fig3.7}.b show both random state exploration and the resulting long-term rewards. The LoS probabilities $\rho_a=0.8$ and $\rho_a=0.2$ are considered, respectively, in Fig. \ref{fig3.7}.a and Fig. \ref{fig3.7}.b. The results in Fig. \ref{fig3.7}.a show that up to $110$ iterations is needed for the QL algorithm to converge, for $\rho_a=0.8$. However, Fig. \ref{fig3.7}.b shows that the algorithm will converge within less than $190$ iterations for $\rho_a=0.2$. Moreover, the average reward is higher in Fig. \ref{fig3.7}.a, since the UAs with $\rho_a=0.8$ are often served over mmW LoS links, while in Fig. \ref{fig3.7}.b, mmW links are frequently blocked. \textcolor{black}{Real-life field measurements have shown that the blockage duration can be very long, exceeding several hundreds of milliseconds \cite{1374945}. This long duration will allow the proposed QL algorithm to converge, before the blockage environment changes.}

In Fig. \ref{fig3.8}.a and Fig. \ref{fig3.8}.b, the average reward and state transitions are shown when the optimal QL policy is followed, respectively, for $\rho_a=0.8$ and $\rho_a=0.2$. Clearly, when LoS probability is high, the optimal policy is to schedule the UA over mmW band, as shown in Fig. \ref{fig3.8}.a. In addition, compared to Fig. \ref{fig3.7}, we can see that the QL policy will substantially increase the average reward compared to the random frequency band selection. For example, for $\rho_a=0.2$, the average reward is increased from $-5$ in Fig. \ref{fig3.7}.b to $1$ in Fig. \ref{fig3.8}.b.
\vspace{-0cm}
\subsection{Outage Probability vs the Required Load}
\begin{figure}
\centering
\centerline{\includegraphics[width=9cm]{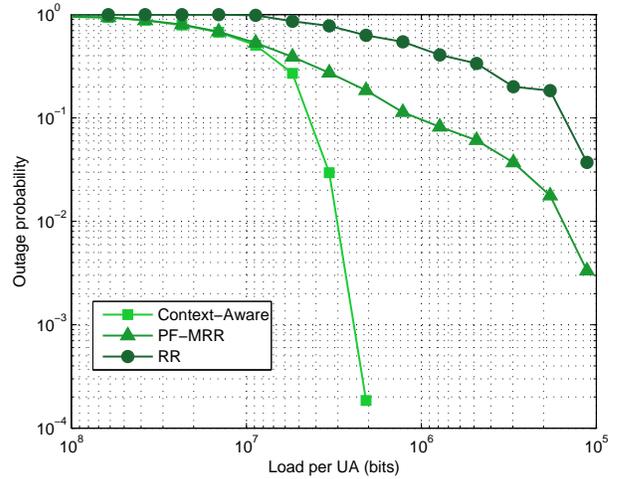}}\vspace{-0.2cm}
\caption{\small \textcolor{black}{Performance comparison between scheduling approaches versus the required bit per UA $b_a$. Parameters $M=30$ and $\rho_a=1$ for all $a \in \mathcal{A}$ are used.}}\vspace{-0cm}
\label{OutagevsBits-1}
\end{figure}

\begin{figure}
\centering
\centerline{\includegraphics[width=9cm]{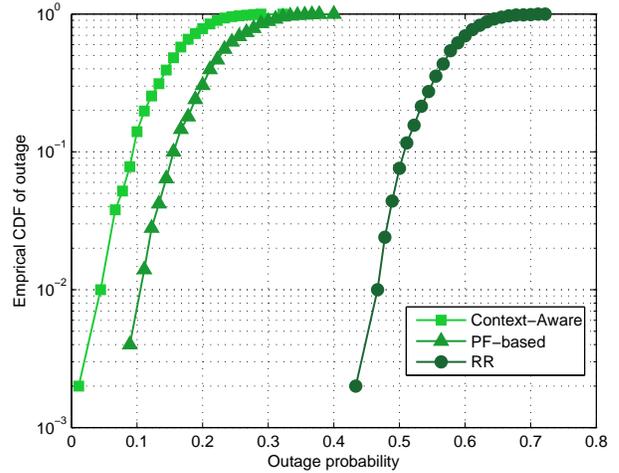}}\vspace{-0.2cm}
\caption{\small \textcolor{black}{The CDF of the outage probability for $M=30$ UAs and $b_a=1$ Mbits. For the cell edge UEs, $\rho_a$ is sampled randomly from $\left[0,1\right]$ and for others $\rho_a=1$.}}\vspace{-0cm}
\label{CDFofOutage-1}
\end{figure}

In Fig. \ref{OutagevsBits-1}, we show the outage probability as the required load per UA varies,  for the three scheduling approaches. In this figure, we can see that the outage probability decreases as the required load per UA decreases. In addition, from Fig. \ref{OutagevsBits-1}, we can see that the proposed context-aware approach yields significant gains, compared to the PF-MRR and RR schemes. \textcolor{black}{In fact, the proposed approach guarantees the required loads up to $2$ Mbits per UA, for $0.01$ outage probability. However, the baseline PF-MRR and RR approaches can guarantee, respectively, less than $0.2$ and $0.1$ Mbits load per UA for the same outage probability.}

\begin{figure*}
\centering
\begin{subfigure}[b]{0.35\textwidth}
\includegraphics[width=\textwidth]{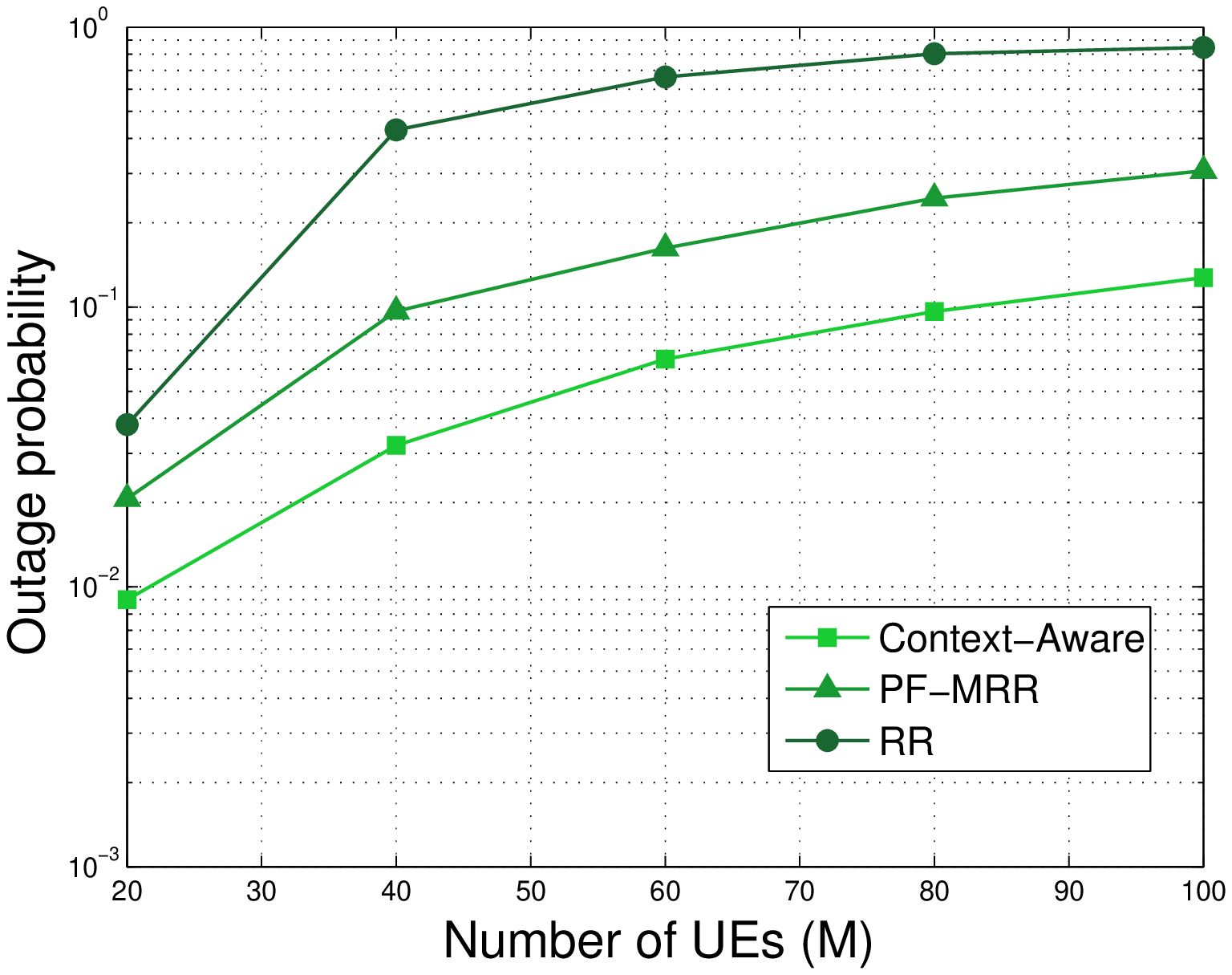}
\caption{Dual-mode}
\label{fig2:a}
\end{subfigure}
~\hspace{-.85cm} 
\begin{subfigure}[b]{0.35\textwidth}
\includegraphics[width=\textwidth]{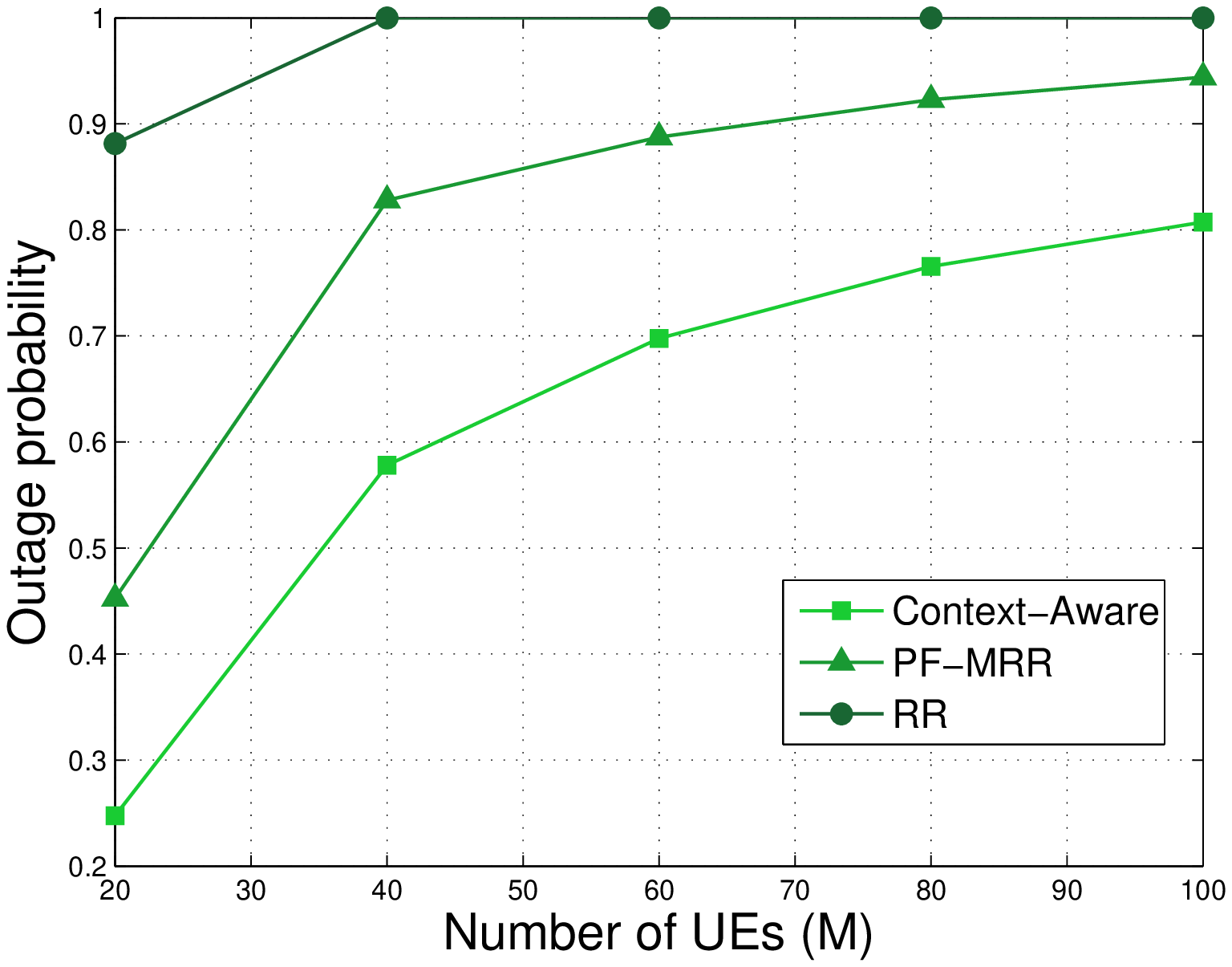}
\caption{Only $\mu$W}
\label{fig2:b}
\end{subfigure}
~\hspace{-.85cm} 
\begin{subfigure}[b]{0.35\textwidth}
\includegraphics[width=\textwidth]{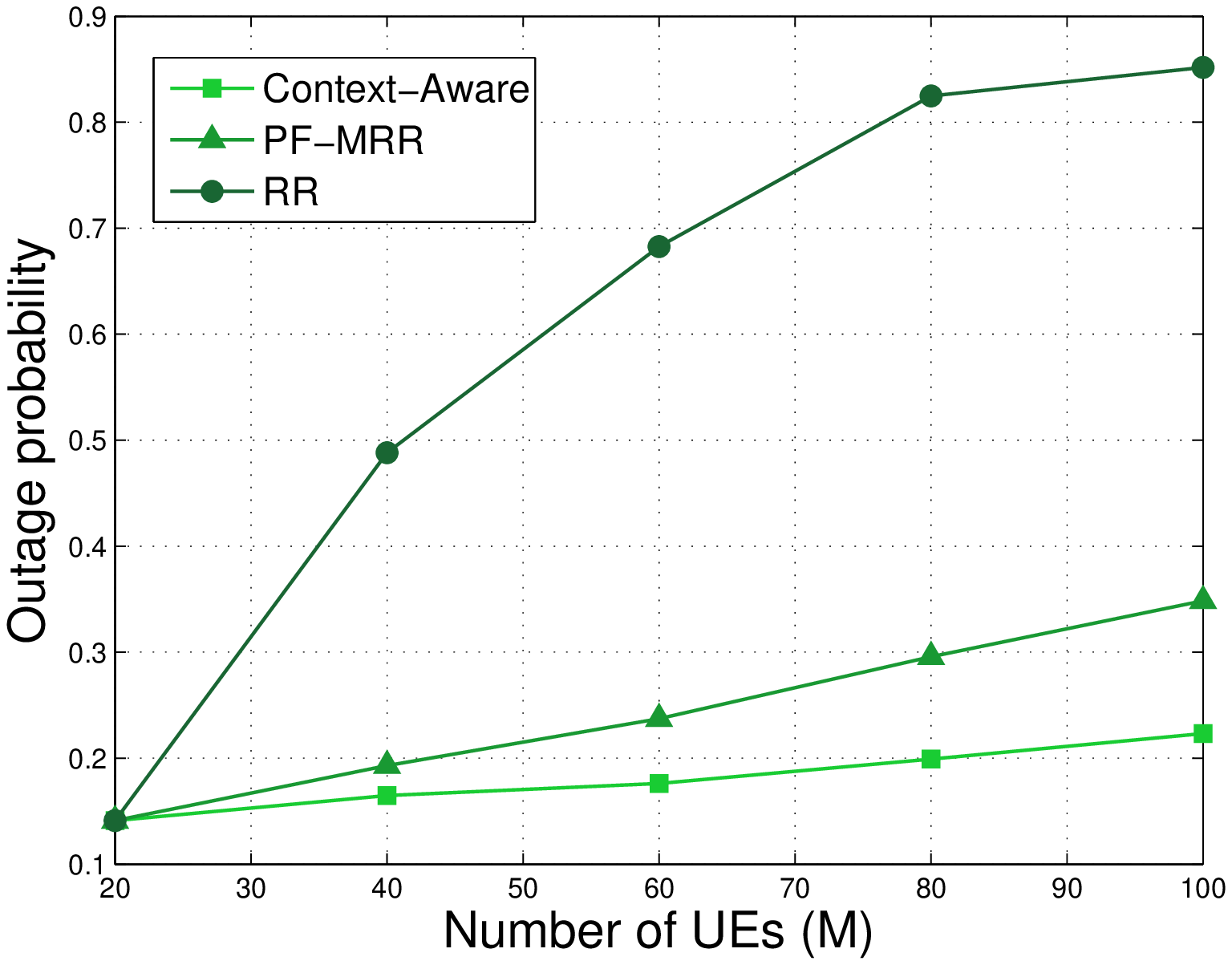}
\caption{Only mmW}
\label{fig2:c}
\end{subfigure}\vspace{-.5em}
\caption{\small \textcolor{black}{Performance comparison between scheduling approaches versus the number of UEs, $M$ for $b_a=0.1$ Mbits. $\rho_a$ is sampled randomly from $\left[0,1\right]$ for the half of UAs.}}\label{Fig2:comment4}\vspace{-1em}
\end{figure*}


\subsection{Statistics of the Outage Probability}

The empirical CDF of the outage probability is shown in Fig. \ref{CDFofOutage-1} for $M=30$ UAs with $b_a=1$ Mbits required load. From Fig. \ref{CDFofOutage-1}, we can see that the proposed context-aware approach substantially improves the statistics of the outage, compared with PF-MRR and RR approaches. \textcolor{black}{For example, the probability that $P_{\text{out}}$ be less than $0.2$ is only $30 \%$ for PF-MRR approach, while this value is $80 \%$ for the proposed approach.}

\vspace{-.2cm}

\subsection{\textcolor{black}{Dual-Mode  vs Single-Mode Scheduling}}

\textcolor{black}{Fig. \ref{Fig2:comment4} shows the performance of the scheduling algorithms for three scenarios: a) with dual-mode communication in presence of both mmW and $\mu$W frequency resources, b) with only $\mu$W band being available, and c) with only mmW band being available\footnote{\textcolor{black}{We note that, in our model, the $\mu$W mode does not employ advanced techniques, such as multi-antenna schemes (e.g., beamforming) or carrier aggregation to achieve higher data rates. Performance evaluation of such advanced $\mu$W systems (e.g. LTE-Advanced) can be considered in future work.}}. The results in Fig. \ref{Fig2:comment4} show the key impact of the proposed dual-mode communication on maximizing QoS, compared with single-mode scenarios. In fact, Fig. \ref{fig2:b} shows that, without mmW communications, the outage probability is significantly high across all network sizes. This is due to the fact that the requested traffic load by UEs falls beyond the available capacity of the network over $\mu$W band. Moreover, Fig. \ref{fig2:c} shows that even for small network sizes, e.g. $M=20$ UEs, the outage probability is greater than $10 \%$ which is significantly high for practical cellular networks. That is because the blockage is likely to happen for the subset of UAs with small $\rho_a$ values. Therefore, to address high traffic loads on the one hand, and guarantee high QoS on the other hand, joint usage of mmW-$\mu$W resources is imperative. Indeed, Fig. \ref{fig2:a} shows that the proposed dual-mode scheduling scheme will yield outage probabilities as low as $1 \%$, while managing very large network sizes up to $300$ UAs, with  a reasonably small outage probability.}
\begin{figure}
\centering
\centerline{\includegraphics[width=9cm]{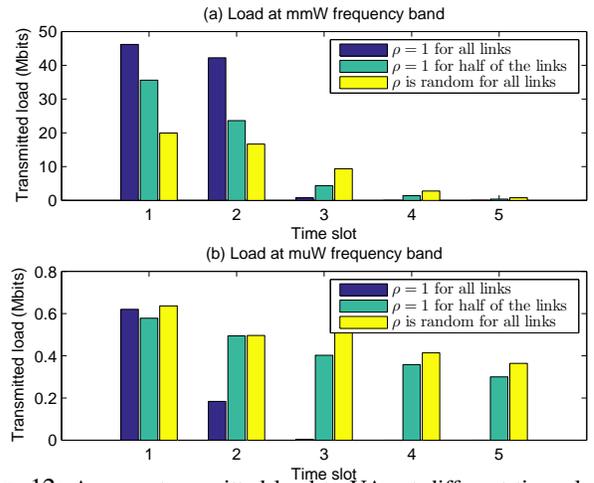}}\vspace{-0.5cm}
\caption{\small \textcolor{black}{Average transmitted load to UAs at different time slots for $A=90$ UAs and $b_a=1$ Mbits.}}
\label{fig3}
\end{figure}

\textcolor{black}{The average transmitted loads to UAs over mmW and $\mu$W frequency bands are shown, respectively, in Figs. \ref{fig3}a and \ref{fig3}b. In fact, Fig. \ref{fig3} shows the average load per radio access technology (RAT) at each time slot. We can observe that the transmitted traffic over the mmW RAT is significantly larger than the $\mu$W RAT. That is clearly due to the larger available bandwidth at the mmW band. Moreover, the transmitted load is lower at last time slots, since by that time, most of the UAs would have already received their requested traffic. In fact, available bandwidth at the mmW band will allow to serve the LoS UAs prior to their due time slot. Clearly, as the link state becomes random for a higher number of UAs, more mmW links will be blocked and, thus, the traffic over the mmW band decreases. Given the results in Figs. \ref{Fig2:comment4} and \ref{fig3}, it is interesting to observe the critical role of exploiting $\mu$W resources, despite the significantly larger traffic at the mmW band. In fact, the joint exploitation of mmW-$\mu$W resources allows to leverage mmW resources for the UAs that are less likely to experience outage, which ultimately decreases traffic at the $\mu$W band in subsequent time slots.}\vspace{-0.2cm}
\subsection{Effect of Beam Training Overhead}
\begin{figure}
\centering
\centerline{\includegraphics[width=9cm]{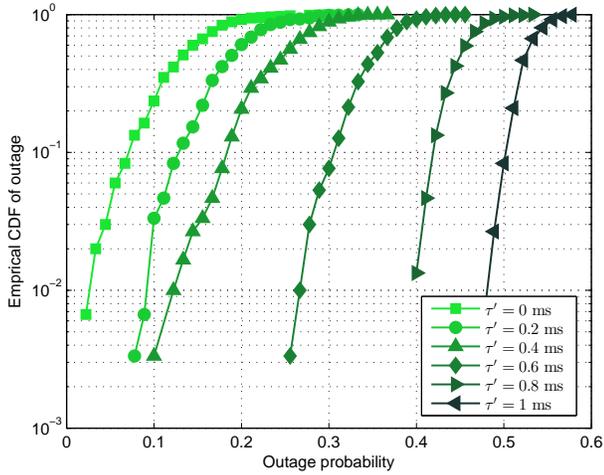}}\vspace{-0.2cm}
\caption{\small \textcolor{black}{The CDF of the outage probability for $M=30$ UAs and $b_a=1$ Mbits. For the cell edge UEs, $\rho_a$ is sampled randomly from $\left[0,1\right]$ and for others $\rho_a=1$.}}\vspace{-0cm}
\label{BF-overhead-1}
\end{figure}

In Fig. \ref{BF-overhead-1}, the effect of the beam training overhead on the outage probability is shown. Here, we observe that $\tau'$ will significantly affect the performance. From Fig. \ref{BF-overhead-1}, we can clearly see that as $\tau'$ increases, the remaining time for data transmissions to UAs decreases which results in a higher outage probability. \textcolor{black}{Fig. \ref{BF-overhead-1} shows that, in the absence of beam training overhead, the outage probability is always less than $0.35$. However, for $\tau'=0.8$ ms, the outage probability will always be less than $0.55$. }
\vspace{-0cm}
\subsection{Number of Iterations}
\begin{figure}
\centering
\centerline{\includegraphics[width=9cm]{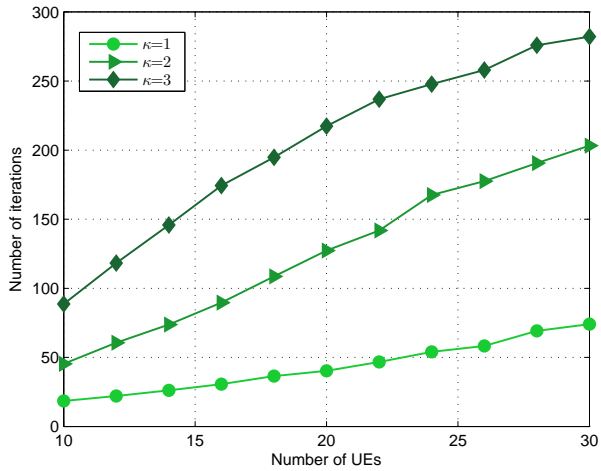}}\vspace{-0.2cm}
\caption{\small The number of iterations versus the number of UEs for $b_a=0.1$ Mbits and $\kappa = 1,2,3$. For the cell edge UEs, $\rho_a$ is sampled randomly from $\left[0,1\right]$ and for others $\rho_a=1$.}\vspace{-0cm}
\label{iteration-1}
\end{figure}

Fig. \ref{iteration-1} shows number of iterations resulting from the proposed scheduling approach as the number of UEs varies for different number of UAs per UE. Clearly, the number of iterations increases almost linearly with the number of UEs. From this figure, we can see that even for large network size up to $30$ UEs and $60$ UAs, the proposed framework is relatively fast, as it converges within $205$ number of iterations.
\vspace{-0cm}
\section{Conclusions}\label{Sec:VI}
In this paper, we have proposed a novel context-aware scheduling framework for dual-mode small base stations operating at mmW and $\mu$W frequency bands. To this end, we have developed a two-stage UA selection and scheduling framework that takes into account various network and UA specific context information to make scheduling decisions. Over the $\mu$W band, we have formulated the context-aware scheduling problem as a one-to-many matching game. To solve this game, we have proposed a novel algorithm for joint UA selection and resource allocation and we have shown that it yields a two-sided stable matching between $\mu$W resources and UAs. Next, we have proposed a joint UA selection and scheduling to allocate mmW resources to the unscheduled UAs. The scheduling problem over mmW band is formulated as a 0-1 Knapsack problem and solved using a suitable algorithm. Moreover, we have proved that the proposed two-stage dual-mode scheduling framework can solve the problem in a polynomial time. Simulation results have shown the various merits and performance advantages of the proposed context-aware scheduling compared to the PF-MRR and RR approaches. 
\vspace{-0.2cm}
\bibliographystyle{IEEEtran}
\bibliography{references}
\begin{IEEEbiography}[{\includegraphics[width=1in,height=1.25in,clip,keepaspectratio]{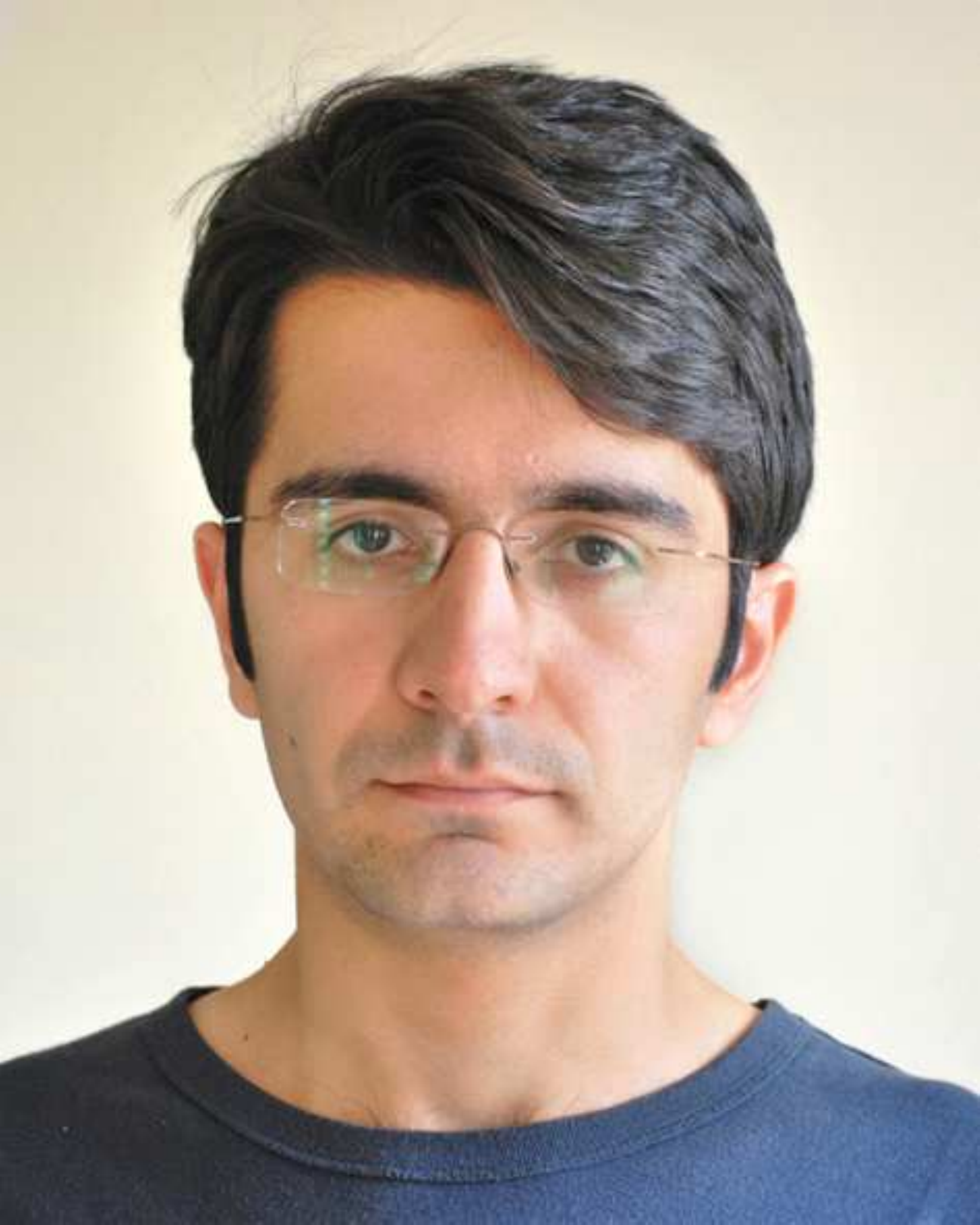}}]{Omid Semiari}(S'14)  received the B.Sc. and M.Sc. degrees in communication systems from University of Tehran in 2010 and 2012, respectively. He is currently a PhD candidate at the Bradly department of Electrical and Computer Engineering at Virginia Tech. In 2014, he has worked as an intern
at Bell Labs, on anticipatory, context-aware resource management. In 2016, he has joined Qualcomm CDMA Technologies (QCT) for a summer internship, working on LTE-Advanced modem design. Mr. Semiari is the recipient of
several research fellowship awards, including DAAD (German Academic Exchange Service) scholarship and NSF
student travel grant. He has actively served as a reviewer for 
flagship IEEE Transactions and conferences and participated as the technical program committee (TPC) member for a variety of workshops at IEEE conferences, such as ICC and GLOBECOM. His research interests
include wireless communications and networking, millimeter wave communications, context-aware
resource allocation, matching theory, and signal processing.\vspace{-1cm}
\end{IEEEbiography}

\begin{IEEEbiography}[{\includegraphics[width=1in,height=1.25in,clip,keepaspectratio]{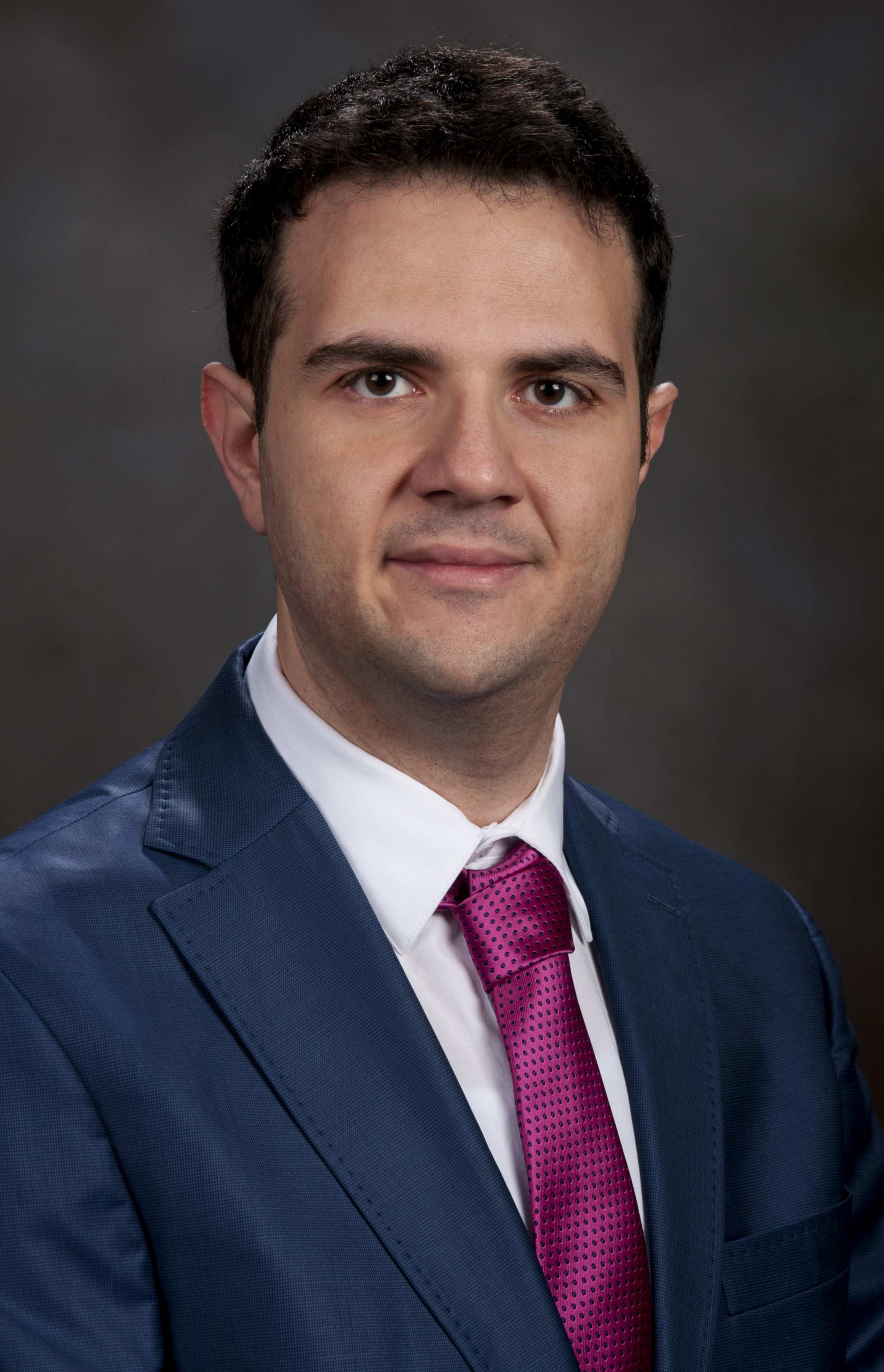}}]{Walid Saad}(S'07, M'10, SM'15) received his Ph.D degree from the University of Oslo in 2010. Currently,  he is an Assistant Professor and the Steven O. Lane Junior Faculty Fellow at the Department of Electrical and Computer Engineering at Virginia Tech, where he leads the Network Science, Wireless, and Security (NetSciWiS) laboratory, within the Wireless@VT research group. His  research interests include wireless networks, game theory, cybersecurity, unmanned aerial vehicles, and cyber-physical systems. Dr. Saad is the recipient of the NSF CAREER award in 2013, the AFOSR summer faculty fellowship in 2014, and the Young Investigator Award from the Office of Naval Research (ONR) in 2015. He was the author/co-author of five conference best paper awards at WiOpt in 2009, ICIMP in 2010, IEEE WCNC in 2012,  IEEE PIMRC in 2015, and IEEE SmartGridComm in 2015. He is the recipient of the 2015 Fred W. Ellersick Prize from the IEEE Communications Society. In 2017, Dr. Saad was named College of Engineering Faculty Fellow at Virginia Tech. Dr. Saad serves as an editor for the IEEE Transactions on Wireless Communications, IEEE Transactions on Communications, and IEEE Transactions on Information Forensics and Security.
\end{IEEEbiography}

\begin{IEEEbiography}[{\includegraphics[width=1in,height=1.25in,clip,keepaspectratio]{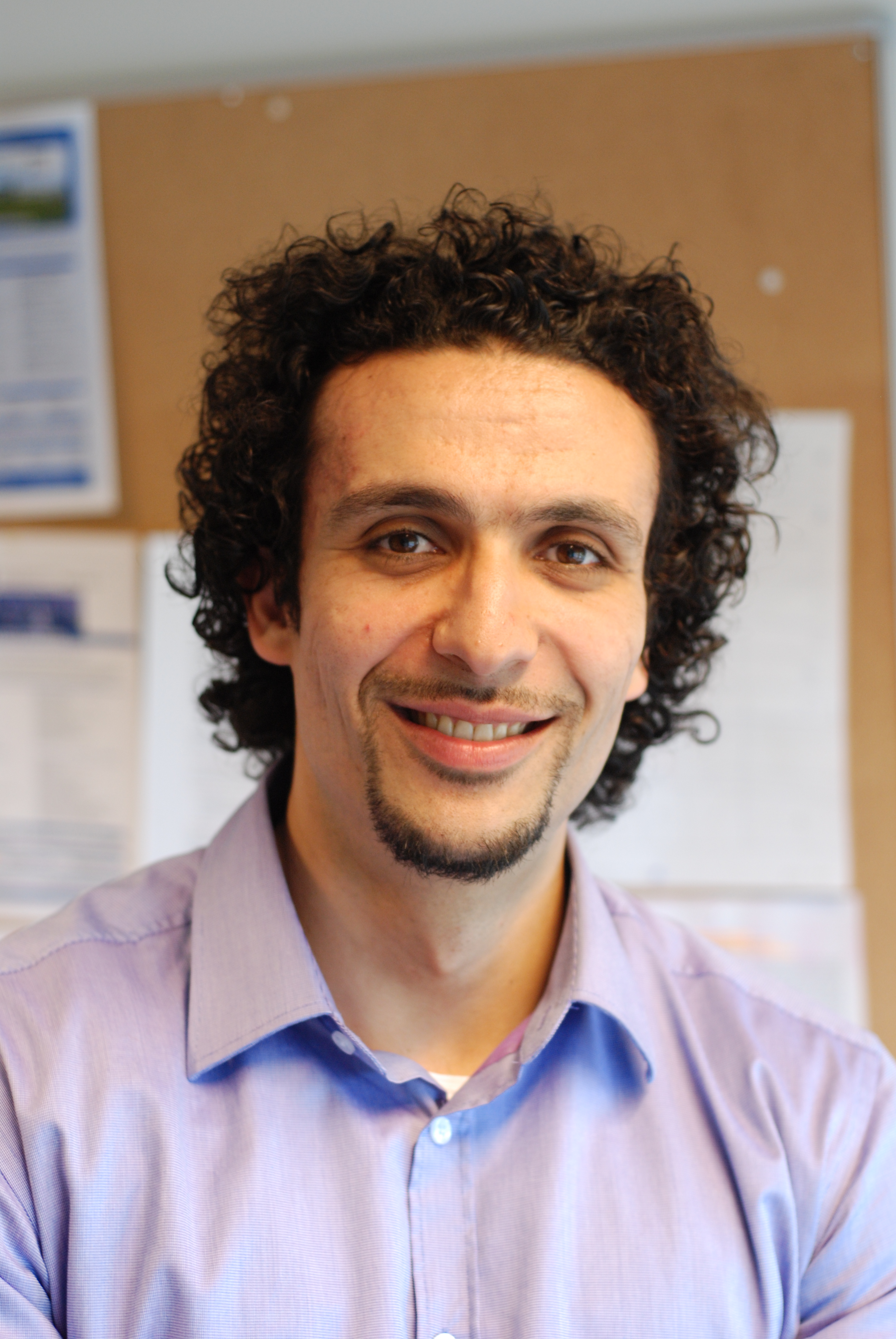}}]{Mehdi Bennis} (Senior Member, IEEE) received his M.Sc. degree in 
Electrical Engineering jointly from the EPFL, Switzerland and the 
Eurecom Institute, France in 2002. From 2002 to 2004, he worked as a 
research engineer at IMRA-EUROPE investigating adaptive equalization 
algorithms for mobile digital
TV. In 2004, he joined the Centre for Wireless Communications (CWC) at 
the University of Oulu, Finland as a research scientist. In 2008, he was 
a visiting researcher at the Alcatel-Lucent chair on flexible radio, 
SUPELEC. He obtained his Ph.D. in December 2009 on spectrum sharing for 
future mobile cellular systems. Currently Dr. Bennis is an Adjunct 
Professor at the University of Oulu and Academy of Finland research 
fellow. His main research interests are in radio resource management, 
heterogeneous networks, game theory and machine learning in 5G networks 
and beyond. He has co-authored one book and published more than 100 
research papers in international conferences, journals and book 
chapters. He was the recipient of the prestigious 2015 Fred W. Ellersick 
Prize from the IEEE Communications Society, the 2016 Best Tutorial 
Prize from the IEEE Communications Society and the 2017 EURASIP Best paper Award for the Journal of wireless communications and networks..
Dr. Bennis serves as an editor for the IEEE Transactions on Wireless Communication
\end{IEEEbiography}

\end{document}